\documentclass[preprint,5p]{elsarticle}

\usepackage[T1]{fontenc}
\usepackage[english]{babel}
\usepackage{lmodern}
\usepackage{graphicx}
\usepackage{amsmath,amsfonts,amsthm}
\usepackage{geometry}
\usepackage[utf8]{inputenc}
\usepackage{multirow}
\usepackage{booktabs}
\usepackage[table]{xcolor}
\usepackage{natbib}
\usepackage{enumitem}

\usepackage{calc}
\newtheorem{theorem}{Theorem}

\newtheorem{corollary}{Corollary}

\biboptions{sort&compress,square}

\journal{ArXiv}

\begin{document}

\renewcommand{\figurename}{Fig.}

\begin{frontmatter}



\title{Assessment of the maintenance cost and analysis of availability measures in a finite life cycle for a system subject to competing failures}

\author[a]{N. C. Caball\'{e}}
\author[b]{I. T. Castro\corref{cor1}}
\ead{inmatorres@unex.es}
\cortext[cor1]{Corresponding author.}

\address[a]{Department of Mathematics, University of Extremadura, Spain}

\begin{abstract}
This paper deals with the assessment of the performance of a system under a finite planning horizon. The
system is subject to two dependent causes of failure: internal degradation and
sudden shocks. We assume that internal degradation follows a gamma process. When the deterioration level of  the degradation
process exceeds a predetermined value, a degradation
failure occurs. Sudden shocks arrive at the system following a doubly stochastic Poisson process (DSPP). A sudden shock provokes the total breakdown of
the system. A condition-based maintenance (CBM) with periodic inspection times is developed. To evaluate the maintenance cost, recursive methods combining numerical integration and Monte Carlo simulation are developed to evalute the expected cost rate and its standard deviation. 
Also, recursive methods to calculate some transient measures of the system are given. Numerical examples are provided to illustrate the analytical results.

\end{abstract}


\end{frontmatter}

\section*{Acronyms}\label{Acronyms}
    \noindent \textit{CBM:}	condition-based maintenance.\\
    \textit{CM:}	corrective maintenance.\\
    \textit{DSPP:}	doubly stochastic Poisson process.\\
    \textit{DTS:}	degradation-threshold-shock.\\
    $d.u.$:		degradation units.\\
    \textit{HPP:} homogeneous Poisson process.\\
    $m.u.$:		monetary units.\\
    \textit{NHPP:}	non-homogeneous Poisson process.\\
    \textit{PM:}	preventive maintenance.\\
    $t.u.$:		time units.\\

\section*{Notation}

     \noindent $A_{T}^{M}(t)$:		availability of the system at time $t$  with time between inspections $T$ and preventive threshold $M$. \vspace{0.2cm}
    
     \noindent $\tilde{A}_{T}^{M}(t)$:		estimation of $A_{T}^{M}(t)$.\vspace{0.2cm}
    
     \noindent $\alpha$, $\beta$:		shape and scale parameters of the gamma process.\vspace{0.2cm}
    
     \noindent $\alpha_{(v_i\%)},\beta_{(v_j\%)}$:  gamma process parameters modified under the $i$-th and $j$-th positions for the vector $\mathbf{v}$.\vspace{0.2cm}
    
     \noindent $C_c$:		corrective maintenance cost.\vspace{0.2cm}
    
     \noindent $C_d$:		downtime cost.\vspace{0.2cm}
    			    				
    \noindent $C_I$:		inspection cost.\vspace{0.2cm}
    
    \noindent $C_p$:		preventive maintenance cost.\vspace{0.2cm}  
    
    \noindent $C(t_1,t_2)$:		maintenance cost in $(t_1,t_2]$.\vspace{0.2cm}
    
    \noindent $C(t_f)$:		maintenance cost in the life cycle.\vspace{0.2cm}
    
    \noindent $C^{\infty}(T,M)$:		expected cost rate for a time between inspections $T$ and preventive threshold $M$.\vspace{0.2cm}
    
    \noindent $D_j$:		length of the $j$-th renewal cycle.\vspace{0.2cm}
    
    \noindent $E\left[C_{T}^{M}(t)\right]$:		expected transient cost at time $t$ with time between inspections $T$ and preventive threshold $M$.\vspace{0.2cm}
    
    \noindent $\tilde{E}\left[C_{T}^{M}(t)\right]$:		estimation of $E\left[C_{T}^{M}(t)\right]$.\vspace{0.2cm}
    
    \noindent $E\left[C_{T}^{M}(t)^2\right]$:		expected square cost at time $t$ with time between inspections $T$ and preventive threshold $M$.\vspace{0.2cm}
    
    \noindent $\tilde{E}\left[C_{T}^{M}(t)^2\right]$:		estimation of $E\left[C_{T}^{M}(t)^2\right]$.\vspace{0.2cm}
   
    \noindent $\tilde{E}^*\left[C^{M}_{T,\alpha_{(v_i\%)},\beta_{(v_j\%)}}(t_f)\right]$:		minimal expected transient cost varying gamma process parameters.\vspace{0.2cm}
   
    \noindent $\tilde{E}^*\left[C^{M}_{T,\lambda_{1,(v_i\%)},\lambda_{2,(v_j\%)}}(t_f)\right]$:		minimal expected transient cost varying sudden shock process parameters.\vspace{0.2cm} 
   
    \noindent $E\left[W_{T}^{M}(t_1,t_2)\right]$:	expected downtime in $(t_1,t_2]$ for a time between inspections $T$ and preventive threshold $M$.\vspace{0.2cm}
    			    				
    \noindent $\tilde{E}\left[W_{T}^{M}(t_1,t_2)\right]$:		estimation of $E\left[W_{T}^{M}(t_1,t_2)\right]$.\vspace{0.2cm}

    \noindent $f_{\alpha t,\beta}$:	density function of the gamma process.\vspace{0.2cm}
    
    \noindent $F_{\sigma_z}$:		distribution function of $\sigma_z$.\vspace{0.2cm}
    
    \noindent $\bar F_{\sigma_{z_2}-\sigma_{z_1}}$:		survival function of $\sigma_{z_2}-\sigma_{z_1}$.\vspace{0.2cm}
    
    \noindent $I(v,t)$:		survival function of $Y$ conditioned to $\sigma_{M_s}=v$.\vspace{0.2cm}
    
    \noindent $IR_T^M(t_1,t_2)$: interval reliability in $(t_1,t_2]$ with time between inspections $T$ and preventive threshold $M$. \vspace{0.2cm}
    
    \noindent $L$:		breakdown threshold.\vspace{0.2cm}
    			 	
    \noindent $\lambda(t)$:		intensity of the sudden shock process.\vspace{0.2cm}
    
    \noindent $\lambda_{1,(v_i\%)},\lambda_{2,(v_j\%)}$:  sudden shock process parameters modified under the $i$-th and $j$-th positions of the vector $\mathbf{v}$.\vspace{0.2cm}
    
    \noindent $M$:	preventive threshold.\vspace{0.2cm}
    
    \noindent $M_s$:	threshold from which the system is more prone to sudden shocks.\vspace{0.2cm}
    
    \noindent $N(t)$:		number of renewals up to $t$.\vspace{0.2cm}
    
    \noindent $N_s(t_1,t_2)$: number of sudden shocks in $(t_1,t_2]$. \vspace{0.2cm}
    
    \noindent $O(t)$:	 	deterioration level of the maintained system at time $t$.\vspace{0.2cm}
     
    \noindent $P_{R_{1}}^{M}(kT)$:		probability of a maintenance action at $k$-th inspection for a time between inspections $T$ and preventive threshold $M$.\vspace{0.2cm}
    			    				
    \noindent $\tilde{P}_{R_{1}}^{M}(kT)$:		estimation of $P_{R_{1}}^{M}(kT)$.\vspace{0.2cm}
    
    \noindent $P_{R_{1,p}}^{M}(kT)$, $P_{R_{1,c}}^{M}(kT)$:		probability of a preventive and corrective maintenance action at $k$-th inspection for a time between inspections $T$ and preventive threshold $M$, respectively.\vspace{0.2cm}
    			    				
    \noindent $\tilde{P}_{R_{1,p}}^{M}(kT)$, $\tilde{P}_{R_{1,c}}^{M}(kT)$:		estimation of $P_{R_{1,p}}^{M}(kT)$ and $P_{R_{1,c}}^{M}(kT)$.\vspace{0.2cm}
    
    \noindent $R_j$:		chronological time of the $j$-th renewal cycle.\vspace{0.2cm}
    
    \noindent $R_{T}^{M}(t)$:		reliability of the system at time $t$ with time between inspections $T$ and preventive threshold $M$.\vspace{0.2cm}
    
    \noindent $\tilde{R}_{T}^{M}(t)$:		estimation of $R_{T}^{M}(t)$.\vspace{0.2cm}
    
    \noindent $S_{T}^{M}(t)$:		standard deviation of the expected cost at time $t$ with time between inspections $T$ and preventive threshold $M$.\vspace{0.2cm}
    
    \noindent $\sigma_z$:	time to reach the deterioration level $z$.\vspace{0.2cm}
    
    \noindent $T$: 	time between inspections.\vspace{0.2cm}
    
    \noindent $t_f$:		length of the life cycle.\vspace{0.2cm}

    \noindent $\mathbf{v}$:		variation vector for model parameters.\vspace{0.2cm}
   
    \noindent $V^{M}_{T,\alpha_{(v_i\%)},\beta_{(v_j\%)}}(t_f)$:		relative variation percentage for gamma process parameters.\vspace{0.2cm}
   
    \noindent $V^{M}_{T,\lambda_{1,(v_i\%)},\lambda_{2,(v_j\%)}}(t_f)$		relative variation percentage for sudden shock process parameters.\vspace{0.2cm}
    
    \noindent $X(t)$:	 	underlying degradation process.\vspace{0.2cm}
    
    \noindent $Y$:		time to a sudden shock.\vspace{0.2cm}

\section{Introduction}
A fundamental aim in the industry field is to ensure the reliability of the systems. It is well-known that some systems suffer a physical degradation process which precedes the failure. This degradation process may involve chemical and physical changes in the system complicating its maintenance. The theory of stochastic processes provides an analytical framework for modelling the impact of the uncertain and time-dependent degradation processes.  

The gamma process is a stochastic cumulative process considered as one of the most appropriated processes for modelling the damage involved by the cumulative deterioration of systems and structures \cite{vanNoortwijk20092}. It is characterised by independent and non-negative gamma increments with identical scale parameters. The gamma process was first applied by Moran \cite{Moran1954} to model water flow into a dam. Later, Abdel-Hammed \cite{5215123} proposed the gamma process as a specific model for deterioration occurring randomly in time. From then on, gamma process has been widely used in the reliability field. The survey by Van Noortwijk  \cite{vanNoortwijk20092} provides many examples of the use of the gamma process in engineering. 

However, some systems are not only subject to internal degradation but are also exposed to sudden shocks which can cause its failure. For example, an ammeter degrades over time and also receive shocks (such as ligthtning) that could provoke the failure of the ammeter \cite{ye2011distribution}. The Light-Emitting Diode (LED) lamps are subject to underlying degradation processes and shocks (over-voltage and over-heating) \cite{SungHo201671}. Even the human body can be regarded as a competing risk system where some health markers can degrade due to the age and some shocks (catastrophic or not) can happen (as a heart attack). Up to our knowledge, Lemoine and Wenocur \cite{NAV:NAV3800320312} were the first to combine both causes of failure, proposing the Degradation-Threshold-Shock (DTS) models. Singpurwalla \cite{Singpurwalla1995} describes diverse failure models for DTS models.
Since two competing failure causes are analyzed, the analysis can be performed considering independent causes of failure \cite{Castro2015} or dependent \cite{Huynh2012140} \cite{Caballe2015RESS},\cite{Cha2016} .

Maintenance strategies regulate the different maintenance tasks which must be performed on the system. Establishing a good maintenance task, the correct functioning of the system is ensured and the maintenance cost can be optimised. Condition-based maintenance (CBM) is one of the most popular techniques used for degrading systems. CBM is a maintenance program that recommends to perform maintenance actions based on the information collected through a condition monitoring process using certain types of sensors or other appropriate indicators \cite{Ahmad2012}. The implementation of CBM programs for DTS models is not new (see \cite{Wenjian20122634} and \cite{Asadzadeh2014117} among others)

Many maintenance designs are planned based on an infinite operating horizon. It means that, after any replacement, the system is renewed by a new one with the same characteristics and the same process is assumed to be repeated indefinitely. Characteristics of these systems, such as the degradation level or the age, are often selected as criteria to optimise the long-run cost rate. Due to renewal properties, this long-run cost rate is equal to the expected cost in a renewal cycle divided by the length of the renewal cycle (see e.g. \cite{Zhao2010921,li2011condition,Fouladirad2011611,Ahmadi2014}). However, most systems actually have a finite operating life cycle since the system cannot always be replaced by a new one with the same characteristics as the previous one an infinite number of times. For instance, in military applications, a missile launching system is only required to be functioning within the designated mission time \cite{Wu2010}. Hence, the use of the asymptotic approach is questionable and the maintenance cost should be analysed under a transient approach. 

Although the transient approach is more realistic than the asymptotic approach, it is less used due to the analytical and computational difficulty of treatment that it involves. However, some works can be found in the literature where authors analyse maintenance strategies under a transient approach for systems subject to competing risks. For example, Taghipour {\it et al.} \cite{taghipour2010} proposed a model to find the optimal interval periodic inspection interval on a finite life cycle for a system subject to different types of failure. 

This paper expands the works by Cheng et al. \cite{Cheng201265} and Pandey et al. \cite{pandey2011} considering the time as a continuous variable and by adding a new component of risk (sudden shocks), whose arrival depends on the degradation process of the system. The framework exposed in \cite{Huynh2011497} inspires the setting of this problem. In this paper, we assume that the system is degraded following a gamma process and sudden shocks arrive at the system following a doubly stochastic Poisson process (DSPP) whose intensity function depends on the degradation process. A CBM with periodical inspection times is developed using the expected cost rate in the life cycle as objective cost function. The evaluation of the maintenance cost in the life cycle of the system is performed using recursive methods. Furthermore, the results obtained using recursive methods are compared to the results obtained based on strictly Monte Carlo simulation. Further comparisons of the maintenance cost are performed considering an infinite life cycle. The robustness of the gamma process parameters and the shock process is also analysed.

In many application fields, there is an increasing interest in evaluating the performance of maintained systems. Reliability and availability are two important performance measures in the traditional reliability field. But some situations are not covered by these indexes and new performance measures, such as the interval reliability, has been developed \cite{Limnios2012,Castro2016}. 
Along with the maintenance assessment, in this paper, some performance measures are also evaluated in the life cycle of the system. The evaluation of these performance measures in the life cycle of the system is performed using recursive methods. Furthermore, the results obtained using recursive methods are compared to the results obtained based on strictly Monte Carlo simulation.

In short, the main contributions of this paper are:
\begin{enumerate}
\item Development of a recursive method to obtain the expected cost and its standard deviation in the life cycle of the system.
\item Comparing the results obtained using this recursive method to the results obtained by using strictly Monte Carlo simulation.
\item Comparing the asymptotic expected cost rate and the expected cost in the life cycle of the system.
\item Analysis of the robustness of some parameters that describe the functioning of the system.
\item Assessment of the the availability, reliability, and interval reliability in the life cycle of the system using a recursive method and comparison of the results obtained by using strictly Monte Carlo simulation.
\end{enumerate} 

In order to develop these contributions, this paper is structured as follows. In Section 2 the general framework of the model is described. Expected transient cost and its standard deviation associated are analysed in Section 3.  Performance measures of the system are exposed in Section 4. Numerical examples are given in Section 5. Conclusions and further possible extensions of this paper are provided in Section 6.

\section{Framework of the problem}
A system subject to two dependent competing causes of failure, degradation and sudden shocks, is considered in this paper. The general assumptions of this model are:
\begin{enumerate}
\item  The system starts working at time $t=0$. This system is subject to an internal degradation process which evolves according to a homogeneous gamma process with parameters $\alpha$ and $\beta,$ where $\alpha, \beta >0$. Let $X(t)$ be the deterioration level of the system at time $t$. Thus, for two time instants  $s$ and $t$, with $s<t$, the density function of the increment deterioration level 
$X(t)-X(s)$ is given by
\begin{equation}\label{densidad_gamma}
f_{\alpha(t-s),\beta}(x)=\frac{\beta^{\alpha(t-s)}}{\Gamma(\alpha(t-s))}x^{\alpha(t-s)-1}e^{-\beta
x},\quad x>0,
\end{equation}
where $\Gamma(\cdot)$ denotes the gamma function defined as
\begin{equation}\label{funcion_gamma}
\Gamma(\alpha)=\displaystyle\int_{0}^{\infty}u^{\alpha-1}e^{-u}~du.
\end{equation}
The system fails due to degradation when the deterioration level exceeds a fixed threshold $L$,
called the breakdown threshold.

\item The system not only fails due to internal degradation, but also it is subject to sudden shocks which cause its failure. Sudden shocks arrive at the system according to a process $\left\{N_s(t),t\geq0\right\}$. This process shows the dependence between degradation and shocks. Following the spirit showed in \cite{Huynh2011497}, we assume that $\left\{N_s(t),t\geq0\right\}$ is a DSPP with intensity $\lambda(t, X(t))$ given by
\begin{equation}\label{intensidad_choques}
\lambda \left(t, X(t)\right)=\lambda_1(t)\mathbf{1}_{\left\{X(t)\leq M_s\right\}}+\lambda_2(t)\mathbf{1}_{\left\{X(t)> M_s\right\}},
\end{equation}
for $t\geq0$,  where $\lambda_1$ and $\lambda_2$ denote two failure rate functions which verify $\lambda_1(t)\leq\lambda_2(t)$, for all $t\geq 0$ and where $\mathbf{1}_{\{\cdot\}}$ denote the indicator function which equals $1$ if the argument is true and $0$ otherwise. The arrival of a sudden shock provokes the system failure.
\item The system is inspected each $T\ (T>0)$ time units ($t.u.$). In these instants, it is checked if the system is working or is down. If the system is down, a corrective maintenance (CM) is performed and the system is replaced by a new one. A CM event simulation is shown in Fig. \ref{m_correctivo}.
\begin{figure}[h]
\centering
\includegraphics[scale=0.20]{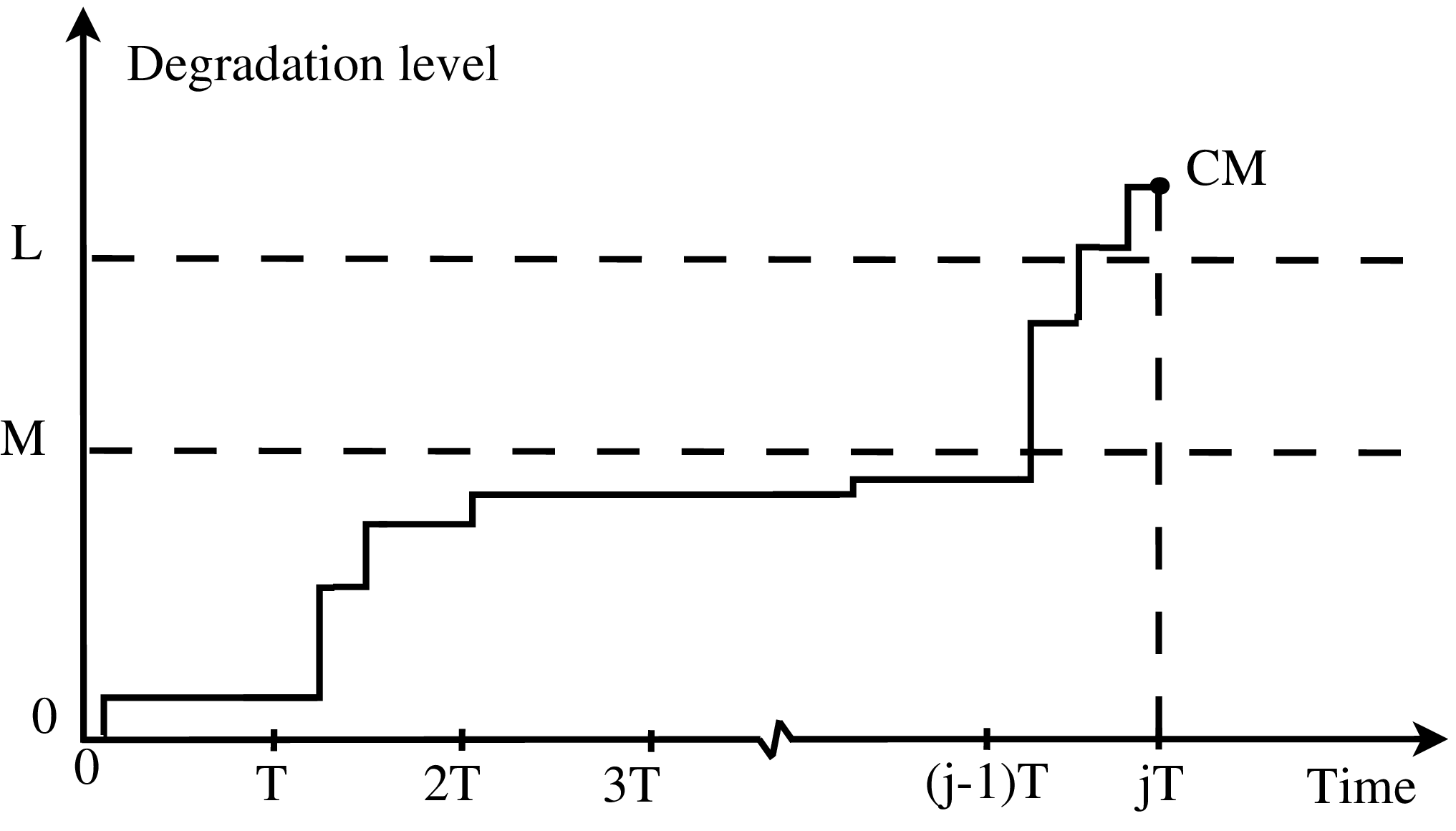}
\caption{A corrective maintenance event.}\label{m_correctivo}
\end{figure}

Let $M$ be the deterioration level from which the system is considered as too worn ($M<L$). If the system is still working and the deterioration level of the system exceeds the preventive threshold $M$, a preventive maintenance (PM) is performed and the system is replaced by a new one. Otherwise, no maintenance action is performed. A PM event simulation is shown in Fig. \ref{m_preventivo}.
\begin{figure}[h]
\centering
\includegraphics[scale=0.20]{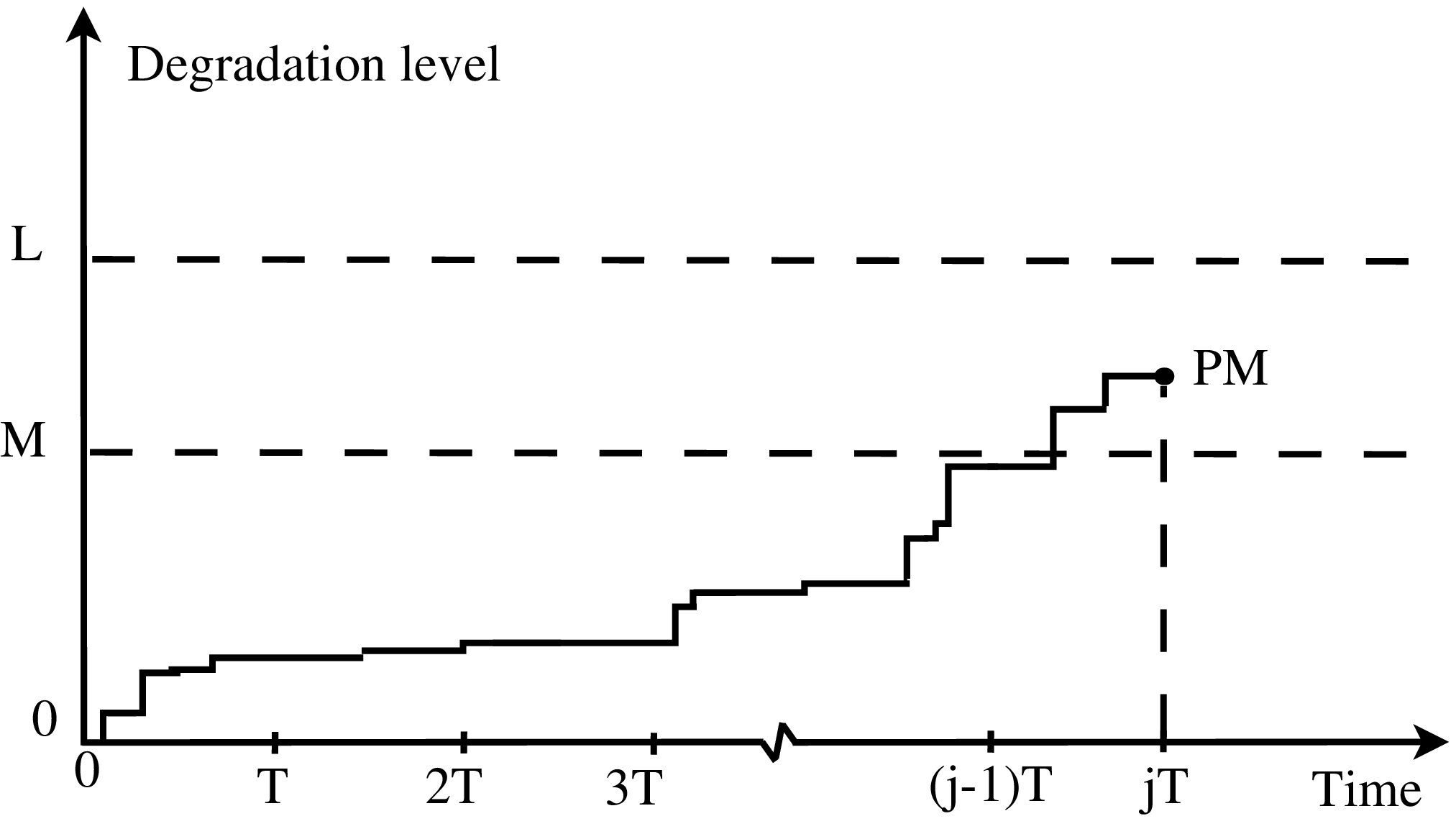}
\caption{A preventive maintenance event.}\label{m_preventivo}
\end{figure}
We assume that the time required to perform a maintenance action is negligible.

\item All maintenance actions imply a cost. A CM and a PM have associated a cost of $C_c$ and $C_p$ monetary units ($m.u.$), respectively, and each inspection implies a cost of $C_I$ $m.u.$ In addition, if the system fails, the system is down until the next inspection. Each time unit that the system is down, a cost of $C_d$ $m.u./t.u.$ is incurred. We assume  $C_c>C_p>C_I$.

\item Let $(0,t_f]$ be the finite operating life cycle of the system. It means that, if the calendar time exceeds $t_f$, the system can no longer be replaced by a new one with the same charasteristics.
\end{enumerate}

Let $\sigma_z$ be the random variable describing the time to reach a certain degradation level $z$. The distribution function of $\sigma_z$, denoted as $F_{\sigma_z}$, is given by
\begin{equation}\label{distribucion_sigmaz}
\begin{array}{l}
\begin{aligned}
F_{\sigma_z}(t)=&P[X(t)\geq
z]\\
=&\displaystyle\int_{z}^{\infty}f_{\alpha
t,\beta}(x)~dx=\frac{\Gamma(\alpha t,z\beta)}{\Gamma(\alpha t)},
\end{aligned}
\end{array}
\end{equation}
for $t\geq0$ where $f_{\alpha t,\beta}(x)$ and $\Gamma(\alpha t)$
are given by (\ref{densidad_gamma}) and (\ref{funcion_gamma}),
respectively, and
\begin{equation*}\label{funcion_gamma_incompleta}
\Gamma(\alpha,x)=\displaystyle\int_{x}^{\infty}u^{\alpha-1}e^{-u}~du,
\end{equation*}
denotes the incomplete gamma function for $x\geq0$ and $\alpha>0$.

For $z_1\leq z_2$, the survival function of $\sigma_{z_2}-\sigma_{z_1}$ is given by
\begin{equation}\label{supervivencia_sigma_z2_sigma_z1}
\begin{array}{l}
\begin{aligned}
\bar F_{\sigma_{z_2}-\sigma_{z_1}}(t)=&P\left[\sigma_{z_2}-\sigma_{z_1}\geq
t\right]\\
=&\displaystyle\int_{x=0}^{\infty}\int_{y=z_1}^{\infty}f_{\sigma_{z_1},X(\sigma_{z_1})}(x,y)\\
&F_{\alpha
t,\beta}(z_2-y)~dy~dx,
\end{aligned}
\end{array}
\end{equation}
where $F_{\alpha t,\beta}$ denotes the distribution function of $f_{\alpha t,\beta}$ and $f_{\sigma_{z_1},X(\sigma_{z_1})}$ denotes the joint density function of $\left(\sigma_{z_1},X(\sigma_{z_1})\right)$ provided by Bertoin
\cite{bertoin1998} as
\begin{equation*}\label{densidad_conjunta}
f_{\sigma_{z_1},X(\sigma_{z_1})}(x,y)=\displaystyle\int_{0}^{\infty}\mathbf{1}_{\left\{z_1\leq
y<z_1+s\right\}} f_{\alpha x,\beta}(y-s)\mu(ds),
\end{equation*}
where $\mu(ds)$ denotes the L{\'e}vy measure associated with a gamma process with parameters $\alpha$ and $\beta$ given by
\begin{equation*}\label{medida_levy}
\mu(ds)=\alpha\frac{e^{-\beta s}}{s},\quad s>0.
\end{equation*}
From Assumption 2, the shock process follows a doubly stochastic Poisson process where the intensity of the shocks depends on time $t$ and on the degradation level $X(t)$. It means that, given a path $x$ de $X(t)$, the process $\left\{N_s(t),t\geq 0\right\}$ is a non homogeneous Poisson process with intensity  $\lambda(t,x)$. In absence of maintenance, the 
time to a failure system is defined as the minimum
$D=\min(\sigma_L,Y)$
where
$$Y=\inf\left(t \geq 0, N_s(t)=1\right), $$
with survival function
$$P(D>t)=\mathbb{E}\left[\mathbf{1}_{\left\{D>t\right\}}exp\left(-\int_{0}^{t}\lambda(s, X(s)ds)\right)\right]$$

Let $I(v,t)$ be the survival function of $Y$ for $t\geq v$, conditioned to  $\sigma_{M_s}=v$. That is
\begin{equation}\label{funcionI}
\begin{array}{l}
\begin{aligned}
I(v,t)=&P\left[Y>t|_{\sigma_{M_s}=v}\right]\\
=&\exp \left\{-\displaystyle\int_0^t \lambda\left(z\right)~dz\right\}
=\frac{\bar F_1(v)}{\bar F_1(0)}\frac{\bar
F_2(t)}{\bar F_2(v)},
\end{aligned}
\end{array}
\end{equation}
where
\begin{equation}\label{supervivenciaF1F2}
\bar F_j(t)=\exp\left\{-\displaystyle\int_0^t\lambda_j(u)du\right\}, \quad j=1,2,
\end{equation}
with density function $f_j(t)$, for $j=1,2$.

\section{Expected transient cost analysis}\label{Section:Expectedcost}
A goal in industry is to find the maintenance strategy that minimises an objective cost function. 
One of the most used function in the literature as objective cost function is the asymptotic cost rate \cite{pandey2011} (asymptotic cost per time unit) that has a simple expression if the functioning of the system can be modelled as a renewal process. 

Since in this paper the system is repaired after each preventive or corrective maintenance, let $D_1, D_2, \ldots$ be the time to the successive renewals of the system. Let $C^{\infty}(T,M)$ be the asymptotic cost rate with a time between inspections $T$ and preventive threshold $M$. Based on the {\it``Renewal Theorem''}, $C^{\infty}(T,M)$ is equal to the expected cost in a renewal cycle divided by the length of the renewal cycle. That is
$$C^{\infty}(T,M)=\lim\limits_{t\rightarrow\infty}\frac{C(t)}{t}=\frac{E[C_1]}{E[D_1]},$$
where $C(t)$ denotes the maintenance cost at time $t$, and $C_1$ and $D_1$ the cost and the length of a renewal cycle, respectively. In this paper, $C^{\infty}(T,M)$ is given by 
\begin{equation}\label{asymptotic_cost_rate}
\begin{array}{l}
\begin{aligned}
&\frac{\displaystyle \sum_{k=1}^{\infty}\Big[C_cP_{R_1,c}^M(kT)+C_pP_{R_1,p}^M(kT)+C_I(k-1)P_{R_1}^M(kT)\Big]}{\displaystyle \sum_{k=1}^{\infty}kT~P_{R_{1}}^M(kT)}\\
+&\frac{\displaystyle \sum_{k=1}^{\infty}C_dE\left[W_T^M((k-1)T,kT)\right]}{\displaystyle \sum_{k=1}^{\infty}kT~P_{R_{1}}^M(kT)},\\
\end{aligned}
\end{array}
\end{equation}
where $P_{R_{1}}^M(kT)$ denotes the probability of the first maintenance action at time $kT$ for $k=1,2,\ldots$ given by
\begin{equation}\label{prob_reemplazamiento}
P_{R_{1}}^M(kT)=P_{R_{1,p}}^M(kT)+P_{R_{1,c}}^M(kT),
\end{equation}
$P_{R_{1,p}}^M(kT)$ and $P_{R_{1,c}}^M(kT)$ denote the probability of the first preventive and corrective maintenance action at time $kT$ for $k=1,2,\ldots$, respectively, and $E\left[W_T^M((k-1)T,kT)\right]$ the expected downtime in $((k-1)T,kT]$. The analytical expressions for these quantities were provided by Huynh {\it et al.} \cite{Huynh2011497}.

Since in this paper the objective is to evaluate the maintenance cost in the life cycle of the system, the expected cost rate in this life cycle is used as objective cost function. 
Let $R_j$ be the chronological time of the $j$-th renewal cycle, for $j=0,1,2,\ldots,N(t_f)+1$, being $N(t_f)$ the number of complete renewals in the life cycle of the system. Thus, $R_j$ is given by
$$R_j=\displaystyle\sum_{n=1}^{j}D_n,$$
for $R_0=0$, where $D_n$ denotes the length of the $n$-th renewal cycle, with $n=1,2,\ldots,N(t_f)+1$. Hence, the length of the $n$-th renewal cycle $D_n$ is given by
\begin{equation*}
D_n=\left\{\begin{array}{ll}
R_n-R_{n-1},\quad &\text{if}\quad n=1,2,\ldots,N(t_f)\\
t_f-R_{N(t_f)},\quad &\text{if}\quad n=N(t_f)+1
\end{array}\right..
\end{equation*}
Fig. \ref{renewal_cycle_sequence} shows a process realisation.
\begin{figure}[h]
\centering
\includegraphics[scale=0.20]{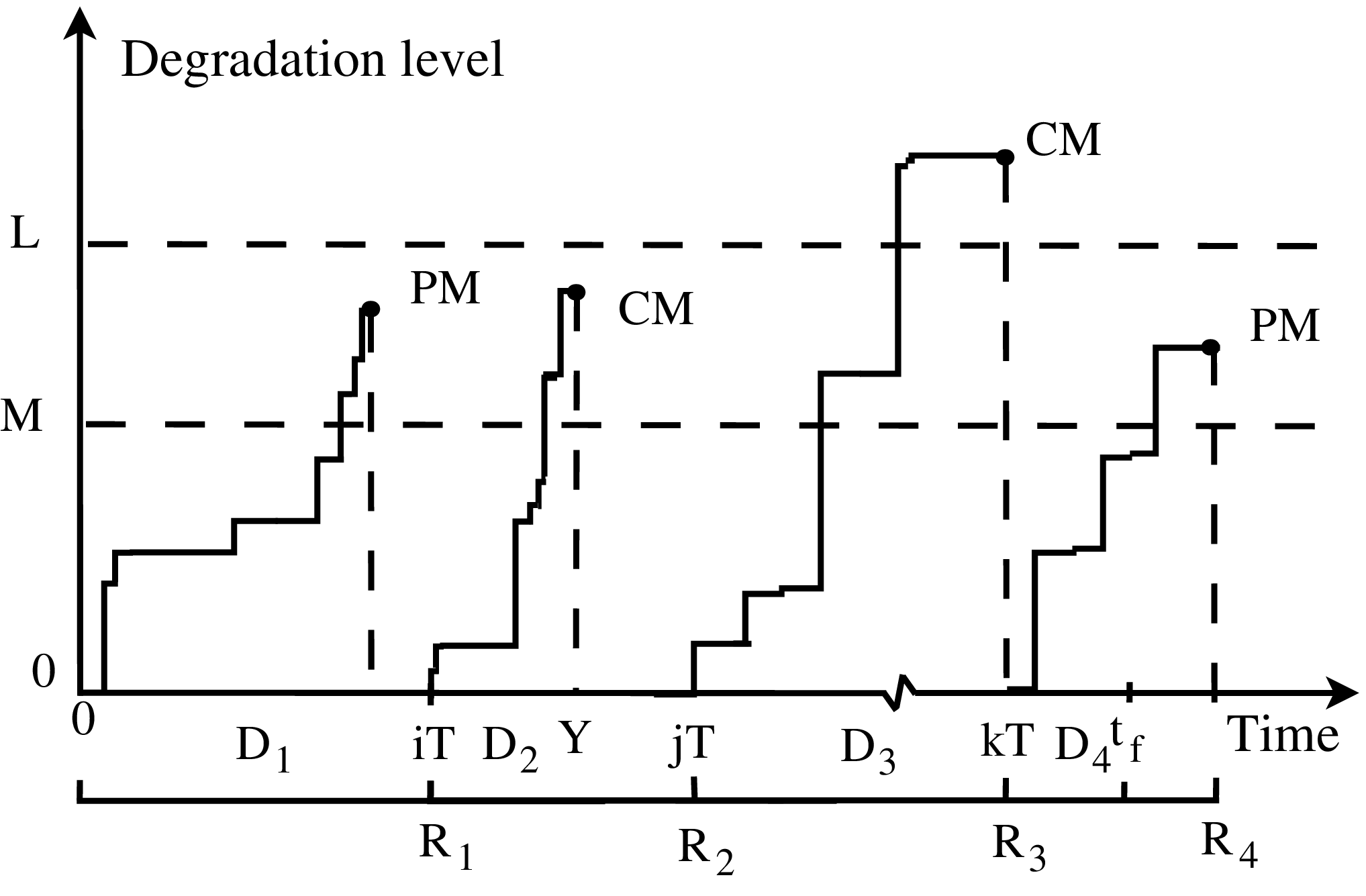}
\caption{A realization of a renewal cycle sequence.}\label{renewal_cycle_sequence}
\end{figure}

The total cost in the finite life cycle is the sum of the incurred costs in the different $N(t_f)$ renewal cycles and the incurred cost in $(R_{N(t_f)},t_f]$. That is
\begin{equation*}\label{coste_total}
C(t_f)=\displaystyle\sum_{j=1}^{N(t_f)} C(R_{j-1},R_j)+C(R_{N(t_f)},t_f),
\end{equation*}
where  $C(t_1,t_2)$ denotes the cost in the interval $(t_1,t_2]$, and $C(0,t_f)$ is simplified as $C(t_f)$.

Let $E\left[C_T^M(t)\right]$ be the expected transient cost at time $t$ with a time between inspections $T$ and a preventive threshold $M$. The next result provides the Markov renewal equation that fulfils $E\left[C_T^M(t)\right]$.

\begin{theorem}\label{theorem_expected_cost}
For $t<T$, the expected transient cost, $E\left[C_T^M(t)\right]$ is given by
\begin{equation*}\label{EC0t}
\begin{array}{l}
\begin{aligned}
E\left[C_T^M(t)\right]=&C_d\displaystyle \int_{0}^{t}f_{\sigma_{M_s}}(u)\int_{u}^{t}\Bigg[-\frac{\partial}{\partial v}\left(I(u,v)\right.\\
&\left.\bar F_{\sigma_{L}-\sigma_{M_s}}(v-u)\right)\Bigg](t-v)~dv~du\\
+&C_d\displaystyle \int_{0}^{t}f_1(u)\bar F_{\sigma_{M_s}}(u)(t-u)du,\\
\end{aligned}
\end{array}
\end{equation*}
where  $\bar F_{\sigma_{M_s}}(u)$ and $f_{\sigma_{M_s}}(u)$ denote the survival function and density function of $\sigma_{M_s}$ given by (\ref{distribucion_sigmaz}), $\bar F_{\sigma_{L}-\sigma_{M_s}}$ and $I(x,y)$ the survival functions given by (\ref{supervivencia_sigma_z2_sigma_z1}) and (\ref{funcionI}), and $f_1(x)$ is the density function of the survival function given by (\ref{supervivenciaF1F2}).

For $t\geq T$, the expected transient cost fulfils the following recursive equation
\begin{equation}\label{ECt}
E\left[C_T^M(t)\right]=\displaystyle \sum_{k=1}^{\lfloor t/T\rfloor}E\left[C_T^M(t-kT)\right] P_{R_1}^M(kT)+G_T^M(t),
\end{equation}
where
\begin{equation*}\label{Gx}
\begin{array}{l}
\begin{aligned}
G_T^M(t)=&\displaystyle \sum_{k=1}^{\lfloor t/T\rfloor}\Big(C_p+C_I(k-1)\Big)P_{R_{1,p}}^M(kT)\\
+&\displaystyle \sum_{k=1}^{\lfloor t/T\rfloor}\Big(C_c+C_I(k-1)\Big)P_{R_{1,c}}^M(kT)\\
+&\displaystyle \sum_{k=1}^{\lfloor t/T\rfloor}C_dE\left[W_T^M((k-1)T,kT)\right]P_{R_{1,c}}^M(kT)\\
+&\lfloor t/T\rfloor C_I\Bigg(1-\sum_{k=1}^{\lfloor t/T\rfloor} P_{R_1}^M(kT)\Bigg)\\
+&C_dE\left[W_T^M(\lfloor t/T\rfloor T,t)\right]\Bigg(1-\sum_{k=1}^{\lfloor t/T\rfloor} P_{R_1}^M(kT)\Bigg),
\end{aligned}
\end{array}
\end{equation*}
with initial condition $E\left[C_T^M(0)\right]=0$ and where $\lfloor t/T \rfloor$ denotes the integer part of $t/T$.
\end{theorem}

\begin{proof}
It is provided in Appendix A.
\end{proof}

\begin{corollary}
Setting $E\left[C_T^{M(i)}(t)\right]=E\left[C_T^{M}(t)\right]$, for all $(i-1)T<t\leq iT$ with $i=1,2,\ldots,\lfloor t_f/T\rfloor$ the expected transient cost, $E\left[C_T^{M(i)}(t)\right]$, is given by
\begin{equation*}
\begin{array}{l}
\begin{aligned}
E\left[C_T^{M(1)}(t)\right]=&C_d\Bigg(\displaystyle \int_{0}^{t}f_{\sigma_{M_s}}(u)\int_{u}^{t}\Bigg[-\frac{\partial}{\partial v}\left(I(u,v)\right.\\
&\left.\bar F_{\sigma_{L}-\sigma_{M_s}}(v-u)\right)\Bigg](t-v)~dv~du\Bigg)\\
+&C_d\displaystyle \int_{0}^{t}f_1(u)\bar F_{\sigma_{M_s}}(u)(t-u)du,\\
\end{aligned}
\end{array}
\end{equation*}
and for $i\geq1$
\begin{equation*}
\begin{array}{l}
\begin{aligned}
E&\left[C_T^{M(i+1)}(t)\right]=G_T^{M(i)}(t)\\
&+\displaystyle \sum_{k=1}^{i}E\left[C_T^{M(i+1-k)}(t-kT)\right] P_{R_1}^M(kT),
\end{aligned}
\end{array}
\end{equation*}
where
\begin{equation*}
\begin{array}{l}
\begin{aligned}
G_T^{M(i)}(t)=&\displaystyle \sum_{k=1}^{i}\Big(C_p+C_I(k-1)\Big)P_{R_{1,p}}^M(kT)\\
+&\displaystyle \sum_{k=1}^{i}\Big(C_c+C_I(k-1)\Big)P_{R_{1,c}}^M(kT)\\
+&\displaystyle \sum_{k=1}^{i}C_dE\left[W_T^M((k-1)T,kT)\right]P_{R_{1,c}}^M(kT)\\
+&iC_I\Big(1-\sum_{k=1}^{i} P_{R_1}^M(kT)\Big)\\
+&C_dE\left[W_T^M(i T,t)\right]\Big(1-\sum_{k=1}^{i} P_{R_1}^M(kT)\Big).
\end{aligned}
\end{array}
\end{equation*}
\end{corollary}

In order to analyse the uncertainty associated with the expected transient cost, the standard deviation is calculated. Let $\Big(S^M_T(t)\Big)^2$ be the variance of expected transient cost at time $t$ with periodic inspection times $T$ and preventive threshold $M$ defined as
\begin{equation}\label{variance_initial}
\Big(S^M_T(t)\Big)^2=E\left[C_T^{M}(t)^2\right]-\left(E\left[C_T^{M}(t)\right]\right)^2.
\end{equation}

Based on Theorem \ref{theorem_expected_cost}, the following result is obtained.

\begin{theorem}\label{Theorem_mean_square}
For $t<T$, the expected square cost at time $t>0$, $E\left[C_T^{M}(t)^2\right]$, is given by
\begin{equation*}\label{B0t}
\begin{array}{l}
\begin{aligned}
E\left[C_T^M(t)^2\right]=&C_d^2\Bigg(\displaystyle \int_{0}^{t}f_{\sigma_{M_s}}(u)\int_{u}^{t}\Bigg[-\frac{\partial}{\partial v}\left(I(u,v)\right.\\
&\left.\bar F_{\sigma_{L}-\sigma_{M_s}}(v-u)\right)\Bigg](t-v)^2~dv~du\Bigg)\\
+&C_d^2\displaystyle \int_{0}^{t}f_1(u)\bar F_{\sigma_{M_s}}(u)(t-u)^2du.\\
\end{aligned}
\end{array}
\end{equation*}

For $t\geq T$, the mean square fulfils the following recursive equation
\begin{equation}\label{Bt}
E\left[C_T^M(t)^2\right]=\displaystyle \sum_{k=1}^{\lfloor t/T\rfloor}E\left[C_T^M(t-kT)^2\right]P_{R_1}^M(kT)+H_T^M(t),
\end{equation}
where
\begin{equation*}\label{H_t}
\begin{array}{l}
\begin{aligned}
H_{T}^M(t)=&\displaystyle \sum_{k=1}^{\lfloor t/T \rfloor } \left(C_p+C_I(k-1)\right)^2P_{R_{1,p}}^M(kT)\\
+&\displaystyle \sum_{k=1}^{\lfloor t/T \rfloor } \left(C_c+C_I(k-1)+C_dE\left[W_{T}^M((k-1)T,kT)\right]\right)^2\\
&P_{R_{1,c}}^M(kT)\\
+&2\sum_{k=1}^{\lfloor t/T \rfloor }\left(C_c+C_I(k-1)\right)E\left[C_T^M(t-kT)\right]\\
&P_{R_{1,c}}^M(kT)\\
+&2\sum_{k=1}^{\lfloor t/T \rfloor }C_dE\left[W_{T}^M((k-1)T,kT)\right]\\
&E\left[C_T^M(t-kT)\right]P_{R_{1,c}}^M(kT)\\
+&2\sum_{k=1}^{\lfloor t/T \rfloor }\left(C_p+C_I(k-1)\right)E\left[C_T^M(t-kT)\right]\\
&P_{R_{1,p}}^M(kT)\\
+&\Big(\lfloor t/T\rfloor C_I +C_dE\left[W_T^M(\lfloor t/T\rfloor T,t)\right] \Big)^2\\
&\Bigg(1-\sum_{k=1}^{\lfloor t/T\rfloor}P_{R_1}^M(kT)\Bigg),\\
\end{aligned}
\end{array}
\end{equation*}
with initial condition $E\left[C_T^M(0)^2\right]=0$.
\end{theorem}

\begin{proof}
It is given in Appendix B.
\end{proof}

\begin{corollary}
Setting $E\left[C_T^{M(i)}(t)^2\right]=E\left[C_T^{M}(t)^2\right]$, for all $(i-1)T<t\leq iT$ with $i=1,2,\ldots,\lfloor t_f/T\rfloor$ the expected square cost, $E\left[C_T^{M(i)}(t)^2\right]$, is given by
\begin{equation*}
\begin{array}{l}
\begin{aligned}
E\left[C_T^{M(1)}(t)^2\right]=&C_d^2\Bigg(\displaystyle \int_{0}^{t}f_{\sigma_{M_s}}(u)\int_{u}^{t}\Bigg[-\frac{\partial}{\partial v}\left(I(u,v)\right.\\
&\left.\bar F_{\sigma_{L}-\sigma_{M_s}}(v-u)\right)\Bigg](t-v)^2~dv~du\Bigg)\\
+&C_d^2\displaystyle \int_{0}^{t}f_1(u)\bar F_{\sigma_{M_s}}(u)(t-u)^2du,\\
\end{aligned}
\end{array}
\end{equation*}
and for $i\geq1$
\begin{equation*}
\begin{array}{l}
\begin{aligned}
E\left[C_T^{M(i+1)}(t)^2\right]=&\displaystyle \sum_{k=1}^{i}E\left[C_T^{M(i+1-k)}(t-kT)^2\right] P_{R_1}^M(kT)\\
+&H_T^{M(i)}(t),
\end{aligned}
\end{array}
\end{equation*}
where
\begin{equation*}
\begin{array}{l}
\begin{aligned}
H_{T}^{M(i)}(t)=&\displaystyle \sum_{k=1}^{i } \left(C_p+C_I(k-1)\right)^2P_{R_{1,p}}^M(kT)\\
+&\displaystyle \sum_{k=1}^{i } \left(C_c+C_I(k-1)+C_dE\left[W_{T}^M((k-1)T,kT)\right]\right)^2\\
&P_{R_{1,c}}^M(kT)\\
+&2\sum_{k=1}^{i }\left(C_c+C_I(k-1)\right)E\left[C_T^M(t-kT)\right]\\
&P_{R_{1,c}}^M(kT)\\
+&2\sum_{k=1}^{i}C_dE\left[W_{T}^M((k-1)T,kT)\right]\\
&E\left[C_T^M(t-kT)\right]P_{R_{1,c}}^M(kT)\\
+&2\sum_{k=1}^{i}\left(C_p+C_I(k-1)\right)E\left[C_T^M(t-kT)\right]\\
&P_{R_{1,p}}^M(kT)\\
+&\Big(i C_I +C_dE\left[W_T^M(i T,t)\right] \Big)^2\\
&\Bigg(1-\sum_{k=1}^{i}P_{R_1}^M(kT)\Bigg).\\
\end{aligned}
\end{array}
\end{equation*}
\end{corollary}

Hence, by (\ref{variance_initial}) the standard deviation of the transient cost at time $t$, $S^M_T(t)$, is given by
\begin{equation}\label{standard_deviation}
S^M_T(t)=\sqrt{E\left[C_T^{M}(t)^2\right]-\left(E\left[C_T^{M}(t)\right]\right)^2}.
\end{equation}

\section{Performance measures of the system}
In addition to the expected cost and its standard deviation associated, recursive expressions for the availability, the reliability and the interval reliability of the system are obtained.

Let $A_T^M(t)$ be the availability of the system at time $t>0$, with time between inspections $T$ and preventive threshold $M$, that is, the probability that the system is working at time $t$

\begin{eqnarray*}
&& A_T^M(t) \\ =
&& \displaystyle\sum_{j=0}^{\infty}\mathbf{1}_{\{R_{j}\leq t<R_{j+1}\}}P\left[X(t-R_j)<L, Y>(t-R_j)\right].
\end{eqnarray*}

Next result provides the Markov renewal equation that fulfils the availability $A_T^M(t)$.

\begin{theorem}\label{availability}
For $t<T$, $A_T^M(t)$ is given by
\begin{equation*}\label{A0t}
\begin{array}{l}
\begin{aligned}
A_T^M(t)=&\bar F_{\sigma_{M_s}}(t)\bar F_{1}(t)\\
+&\displaystyle \int_{0}^{t}f_{\sigma_{M_s}}(u)\bar F_{\sigma_{L}-\sigma_{M_s}}(t-u)I(u,t)du.\\
\end{aligned}
\end{array}
\end{equation*}
For $t\geq T$, $A_T^M(t)$ fulfils the following Markov renewal equation
\begin{equation}\label{At}
\begin{array}{l}
\begin{aligned}
A_T^M(t)=&\displaystyle \sum_{k=1}^{\lfloor t/T\rfloor}A_T^M(t-kT)P_{R_1}^M(kT)\\
+&J_{T,1}^M(t)\mathbf{1}_{\left\{M\leq M_s\right\}}+J_{T,2}^M(t)\mathbf{1}_{\left\{M> M_s\right\}},
\end{aligned}
\end{array}
\end{equation}
where
\begin{equation}\label{Jt1}
\begin{array}{l}
\begin{aligned}
J_{T,1}^M(t)=&\bar F_{\sigma_{M}}(t)\bar F_{1}(t)\\
+&\displaystyle \int_{\lfloor t/T \rfloor T}^t f_{\sigma_{M}}(u)\int_{u}^{t}f_{\sigma_{M_s}-\sigma_{M}}(v-u)\\
&\bar F_{\sigma_{L}-\sigma_{M_s}}(t-v)I(v,t)~dv~du\\
+&\displaystyle \int_{\lfloor t/T \rfloor T}^t f_{\sigma_{M}}(u)\bar F_{\sigma_{M_s}-\sigma_{M}}(t-u)\bar F_{1}(t)~du,\\
\end{aligned}
\end{array}
\end{equation}
and 
\begin{equation}\label{Jt2}
\begin{array}{l}
\begin{aligned}
J_{T,2}^M(t)=&\bar F_{\sigma_{M_s}}(t)\bar F_{1}(t)\\
+&\displaystyle \int_{0}^{\lfloor t/T \rfloor T} f_{\sigma_{M_s}}(u)\int_{\lfloor t/T \rfloor T}^{t}f_{\sigma_{M}-\sigma_{M_s}}(v-u)\\
&\bar F_{\sigma_{L}-\sigma_{M}}(t-v)I(u,t)~dv~du\\
+&\displaystyle \int_{0}^{\lfloor t/T \rfloor T} f_{\sigma_{M_s}}(u)\bar F_{\sigma_{M}-\sigma_{M_s}}(t-u)I(u,t)~du\\
+&\displaystyle \int_{\lfloor t/T \rfloor T}^t f_{\sigma_{M_s}}(u)\bar F_{\sigma_{L}-\sigma_{M_s}}(t-u)I(u,t)~du,\\
\end{aligned}
\end{array}
\end{equation}
with initial condition $A_T^M(0)=1$ and where $P_{R_1}^M(kT)$ is given by (\ref{prob_reemplazamiento}).
\end{theorem}

\begin{proof}
It is given in Appendix C.
\end{proof}

Often, it is also of interest the probability that the system starts working at time 0 and it continues operating for a time interval. 
Let $R_T^M(t)$ be the reliability of the system at time $t$ with time between inspections $T$ and preventive threshold $M$, that is, the probability that the system is working in $(0,t]$ given by
\begin{equation*}
\begin{array}{l}
\begin{aligned}
R_T^M(t)=& P\left[O(u)<L, \forall u \in (0,t], N_s(0,t)=0\right], 
\end{aligned}
\end{array}
\end{equation*}
where $O(t)$ denotes the deterioration level of the maintained system at time $t$, that is, 
$$O(t)=\sum_{j=0}^{\infty}\mathbf{1}_{\{R_{j}\leq t<R_{j+1}\}} X(t-R_j)$$

Based on Theorem \ref{availability}, the following result is obtained.

\begin{theorem}\label{reliability}
For $t<T$, $R_T^M(t)$ is given by 
\begin{equation*}\label{R0t}
\begin{array}{l}
\begin{aligned}
R_T^M(t)=&A_T^M(t).\\
\end{aligned}
\end{array}
\end{equation*}
For $t\geq T$, $R_T^M(t)$ fulfils the following Markov renewal equation
\begin{equation}\label{Rt}
\begin{array}{l}
\begin{aligned}
R_T^M(t)=&\displaystyle \sum_{k=1}^{\lfloor t/T\rfloor}R_T^M(t-kT)P_{R_{1,p}}^M(kT)\\
+&J_{T,1}^M(t)\mathbf{1}_{\left\{M\leq M_s\right\}}+J_{T,2}^M(t)\mathbf{1}_{\left\{M> M_s\right\}},
\end{aligned}
\end{array}
\end{equation}
with initial condition $R_T^M(0)=1$ where $J_{T,1}^M$ and $J_{T,2}^M$ are given by (\ref{Jt1}) and (\ref{Jt2}), respectively.
\end{theorem}

\begin{proof}
It is given in Appendix D.
\end{proof}

A performance measure that extends the reliability is the interval reliability, defined as the probability that the system is working at time $t$, and will continue working over a finite time interval of length $s$. The joint interval reliability is applied when there are periods in the lifetime cycle when a failure should be avoided with high probability. Let $IR_T^M(t,t+s)$ be the interval reliability in $(t,t+s]$. That is
\begin{eqnarray*}
&& IR_T^M(t,t+s) \\
&& =P\left[O(u)<L, \forall u \in (t,t+s], N_s(t,t+s)=0\right].
\end{eqnarray*}
Availability and reliability are particular cases of the interval reliability since 
$$A_T^M(t)=IR_T^M(t,t+0),$$
and
$$R_T^M(t)=IR_T^M(0,0+t).$$

Based on Theorems \ref{availability} and \ref{reliability}, the following result is obtained.

\begin{theorem}\label{reliabilityinterval}
For $t+s<T$, $IR_T^M(t,t+s)$ is given by
\begin{equation*}\label{IR0t}
\begin{array}{l}
\begin{aligned}
IR_T^M(t,t+s)=&R_T^M(t+s).\\
\end{aligned}
\end{array}
\end{equation*}
For $t+s\geq T$, $IR_T^M(t,t+s)$ fulfils the following Markov renewal equation
\begin{equation}\label{IRt}
\begin{array}{l}
\begin{aligned}
IR_T^M(t,t+s)=&\displaystyle \sum_{k=\lfloor t/T\rfloor+1}^{\lfloor (t+s)/T\rfloor}R_T^M(t+s-kT)P_{R_{1,p}}^M(kT)\\
+&\displaystyle \sum_{k=1}^{\lfloor t/T\rfloor}IR_T^M(t-kT,t+s-kT)P_{R_{1}}^M(kT)\\
+&J_{T,1}^M(t+s)\mathbf{1}_{\left\{M\leq M_s\right\}}+J_{T,2}^M(t+s)\mathbf{1}_{\left\{M> M_s\right\}},
\end{aligned}
\end{array}
\end{equation}
with initial conditions $R_T^M(0)=1$ and $IR_T^M(0,0)=1$.
\end{theorem}

\begin{proof}
It is given in Appendix E.
\end{proof}

\section{Numerical examples}
In this section, some numerical examples are provided to illustrate the analytical results. To this end, we consider a system subject to an underlying degradation process modelled as a homogeneous gamma process with parameters $\alpha=\beta=0.1$. We assume that the system fails when the deterioration level of the system reaches the breakdown threshold $L=30$. The system can also fail due to a sudden shock and the sudden shock process is modelled under a DSPP with intensity
$$\lambda\left(t, X(t)\right)=0.01\cdot\mathbf{1}_{\left\{X(t)\leq M_s\right\}}+0.1\cdot\mathbf{1}_{\left\{X(t)> M_s\right\}},\quad t\geq0, 
$$
where $M_s=20$. 
Under these conditions, the expected time to the failure due to deterioration is $E\left[\sigma_L\right]=34.0335$ $t.u.$ and the expected time to the failure due to a sudden shock is $E\left[Y\right]=28.3556$ $t.u.$
In addition, we assume the cost sequence $C_c=300$ $m.u.$, $C_p=150$ $m.u.$, $C_I=45$ $m.u.$, and $C_d=25$ $m.u./t.u.$ We assume that the life cycle of the system is $(0,50]$.  

The following examples were executed by using MATLAB software version R2014a on an Intel Core i5-2500 processor with 8GB DDR3 RAM, running under Windows 7 Professional.
\subsection{Asymptotic expected cost rate analysis}\label{two-dimensional_aymptotic}
Considering the previous dataset, the expected cost rate is first computed to establish the values of $T$ and $M$ which will be used subsequently in the transient approach analysis. In this way, the results obtained by using asymptotic and transient approaches shall be compared.

The optimisation problem for the expected cost based on the asymptotic formula given in (\ref{asymptotic_cost_rate}) is computed as follows:

\begin{enumerate}
\item A grid of size $10$ is obtained discretising the set $[5,50]$ into $10$ equally spaced points from $5$ to $50$ for $T$. For $i=1,2,\ldots,10$, let $T_i$ be the $i$-th value of the grid obtained previously.
\item A grid of size $30$ is obtained by discretising the set $[1,30]$ into $30$ equally spaced points from $1$ to $30$ for $M$. For $j=1,2,\ldots,30$, let $M_j$ be the $j$-th value of the grid obtained previously.
\item For each fixed combination $(T_i,M_j)$, we obtain $50000$ simulations of $(R_1,I_1,W_d)$, where $R_1$ corresponds to the time to the first replacement (corrective or preventive), $I_1$ the nature of the first maintenance action performed (corrective or preventive), and $W_d$ the downtime up to the first maintenance action. 

With these simulations, and applying Monte Carlo method, we obtain $\tilde{P}_{R_{1,p}}^{M_j}(kT_i)$, $\tilde{P}_{R_{1,c}}^{M_j}(kT_i)$, $\tilde{P}_{R_1}^{M_j}(kT_i)$, and $\tilde{E}\left[W_{T_i}^{M_j}((k-1)T_i,kT_i))\right]$ which correspond to the estimations of $P_{R_{1,p}}^{M_j}(kT_i))$, $P_{R_{1,c}}^{M_j}(kT_i))$, $P_{R_1}^{M_j}(kT_i))$, and $E\left[W_{T_i}^{M_j}((k-1)T_i,kT_i))\right]$ for $k=1,2,\ldots,\lfloor 50/T_i\rfloor$, respectively (see Fig. \ref{flujos_variando_T_M}).
\begin{figure}[h]
\begin{center}
\includegraphics[scale=0.45]{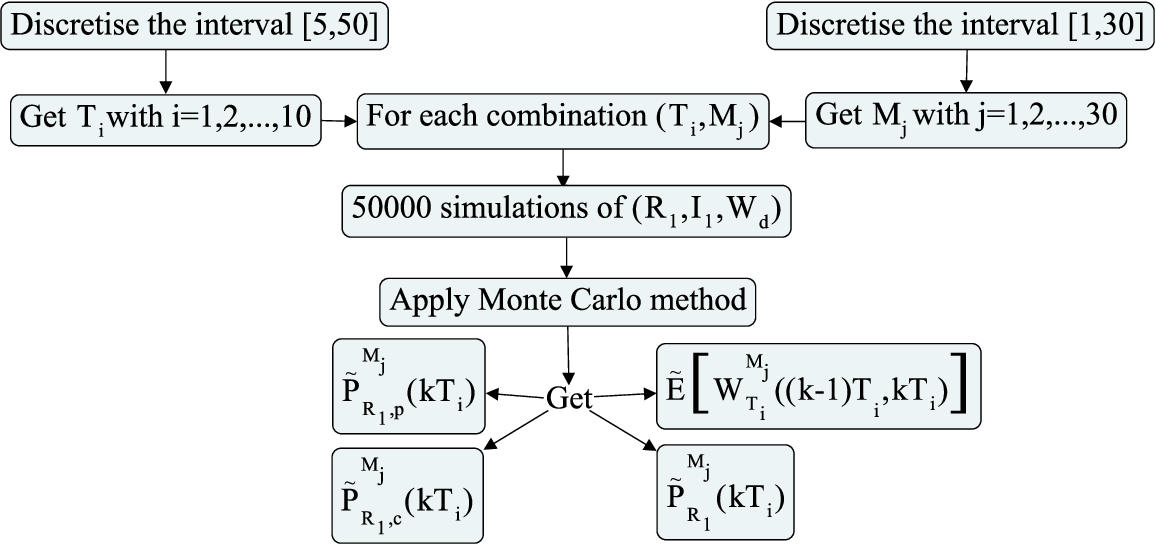}
\caption{Procedure of the Monte Carlo simulation method for variable $T$ and $M$.}\label{flujos_variando_T_M}
\end{center}
\end{figure}

\item Quantity $\tilde{C}^{\infty}(T,M)$, which represents the asymptotic expected cost rate, is calculated by using Equation (\ref{asymptotic_cost_rate}) replacing the corresponding probabilities by their estimations calculated in Step 3.

\item The optimisation problem is reduced to find the values $T_{opt}$ and $M_{opt}$ which minimise the asymptotic expected cost rate $\tilde{C}^{\infty}(T,M)$. That is 
\begin{equation*}
\tilde{C}^{\infty}(T_{opt},M_{opt})=\min_{\substack{T\geq 0 \\ 0\leq M\leq L}}\left\{\tilde{C}^{\infty}(T,M)\right\}.
\end{equation*}

\end{enumerate}

\begin{figure*}
\centering
\begin{minipage}[H]{.45\textwidth}
\begin{center}
\includegraphics[scale=0.45]{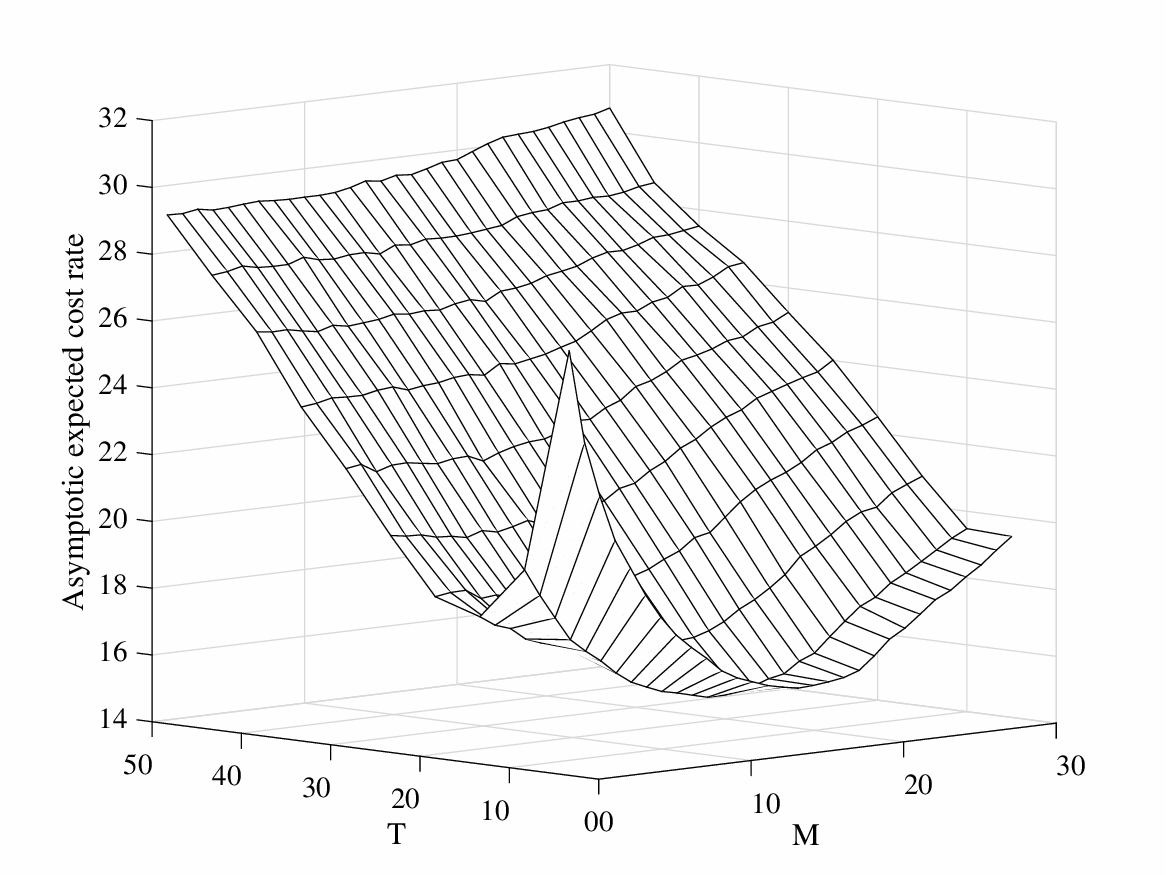}
\end{center}
\end{minipage}
\hfill
\begin{minipage}[H]{.45\textwidth}
\begin{center}
\includegraphics[scale=0.45]{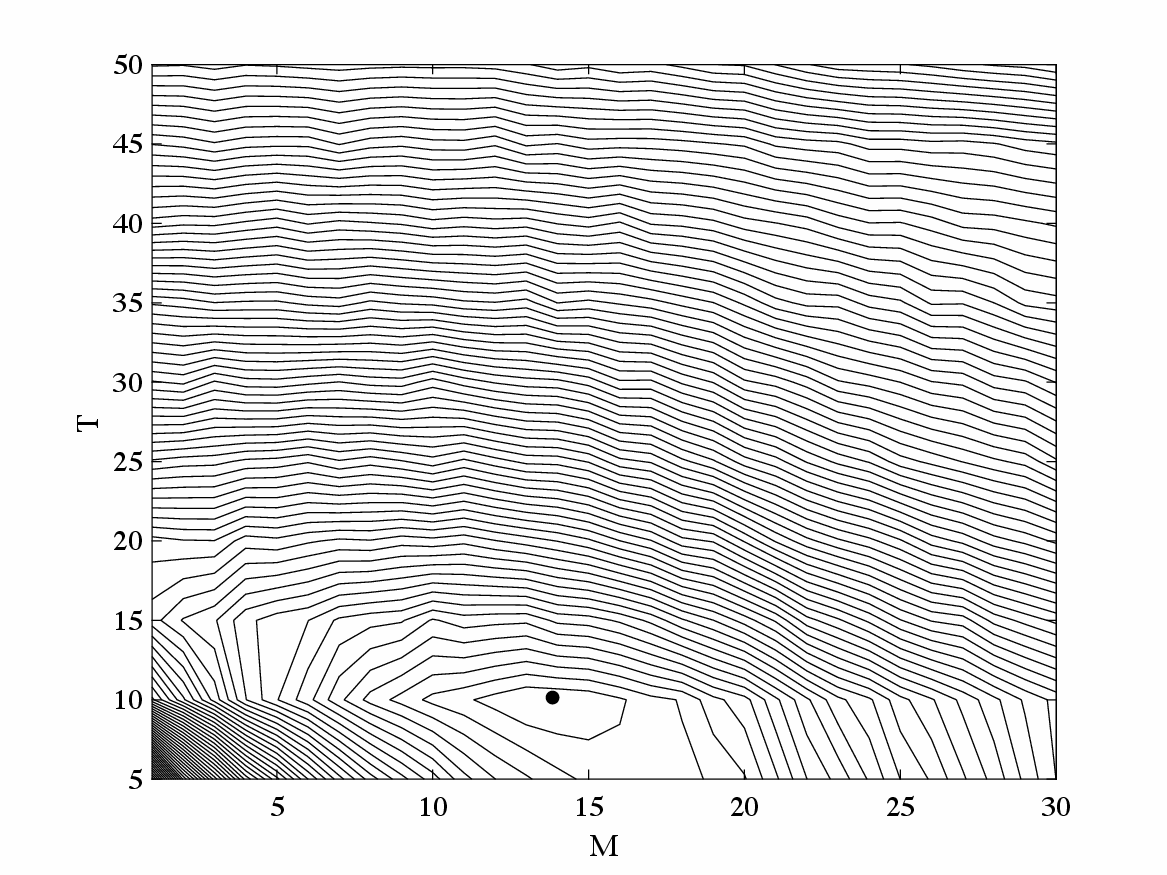}
\end{center}
\end{minipage}
\caption{\label{contour_mesh_asymptotic_cost_rate}Mesh and contour plots for the asymptotic expected cost rate.}
\end{figure*}

Fig. \ref{contour_mesh_asymptotic_cost_rate} shows the expected cost rate versus $T$ and $M$. The values of $T$ and $M$ which minimise the expected cost rate are reached at $M_{opt}=14$ $d.u.$ and $T_{opt}=10$ $t.u.$, with an expected cost rate of $15.3819$ $m.u./t.u.$ Below, the expected cost in the finite life cycle will be compared to the asymptotic expected cost using the values $T_{opt}$ and $M_{opt}$.

\subsection{Expected transient cost rate analysis for a fixed $T$}\label{Section_T_fixed}
We consider a time between inspections $T=10$ $t.u.$ The optimisation problem for the expected transient cost based on the recursive formula given in (\ref{ECt}) is computed as follows:
\begin{enumerate}\label{procedimiento_recursivo_variandoM}

\item A grid of size $30$ is obtained by discretising the set $[1,30]$ into $30$ equally spaced points from $1$ to $30$ for $M$. 

\item For fixed $T=10$ $t.u.$ and for each fixed $M_j$, we obtain $50000$ simulations of $(R_1,I_1,W_d)$.  With these simulations, and applying Monte Carlo method, we obtain the estimations $\tilde{P}_{R_{1,p}}^{M_j}(10k)$, $\tilde{P}_{R_{1,c}}^{M_j}(10k)$, $\tilde{P}_{R_1}^{M_j}(10k)$, and $\tilde{E}\left[W_{10}^{M_j}((k-1)10,10k)\right]$ (see Fig. \ref{flujos_variando_T_M}).

\item For fixed $T=10$ $t.u.$ and $M_j$, let $\tilde{E}\left[C_{10}^{M_j}(50)\right]$ be the expected cost in the finite life cycle. The expected transient cost is calculated by using the recursive formula given in (\ref{ECt}), replacing the corresponding probabilities by their estimations calculated in Step 2, with initial condition $\tilde{E}\left[C_{10}^{M_j}(0)\right]=0$. 

\item For fixed $T=10$ $t.u.$, the optimisation problem is reduced to find the value $M_{opt}$ which minimises the expected cost $\tilde{E}\left[C_{10}^{M}(50)\right]$. That is
\begin{equation*}
\tilde{E}\left[C_{10}^{M_{opt}}(50)\right]=\min_{0\leq M\leq L}\left\{\tilde{E}\left[C_{10}^{M}(50)\right]\right\}.
\end{equation*} 

\end{enumerate}

Let $\tilde{E}\left[C_{10}^{M_j}(50)\right]/50$ be the expected cost rate in the life cycle of the system. The expected cost rate calculated using the recursive method and the expected cost rate calculated using strictly Monte Carlo simulation are shown in Fig. \ref{costes_generalM}. The expected transient cost rate based on strictly Monte Carlo simulation was calculated for $30$ equally spaced points in the interval $(0,30]$ with $50000$ simulations for each point. Based on Fig. \ref{costes_generalM}, for the recursive method, the expected transient cost rate based on the recursive method given by Theorem \ref{theorem_expected_cost} reaches its minimum value at $M=14$ $d.u.$, with an expected transient cost rate of $14.7639$ $m.u./t.u.$ On the other hand, the expected transient cost rate based on strictly Monte Carlo simulation reaches its minimum value at $M=14$ $d.u.$, with an expected transient cost rate of $15.2096$ $m.u./t.u.$
\begin{figure}[h]
\centering
\includegraphics[scale=0.30]{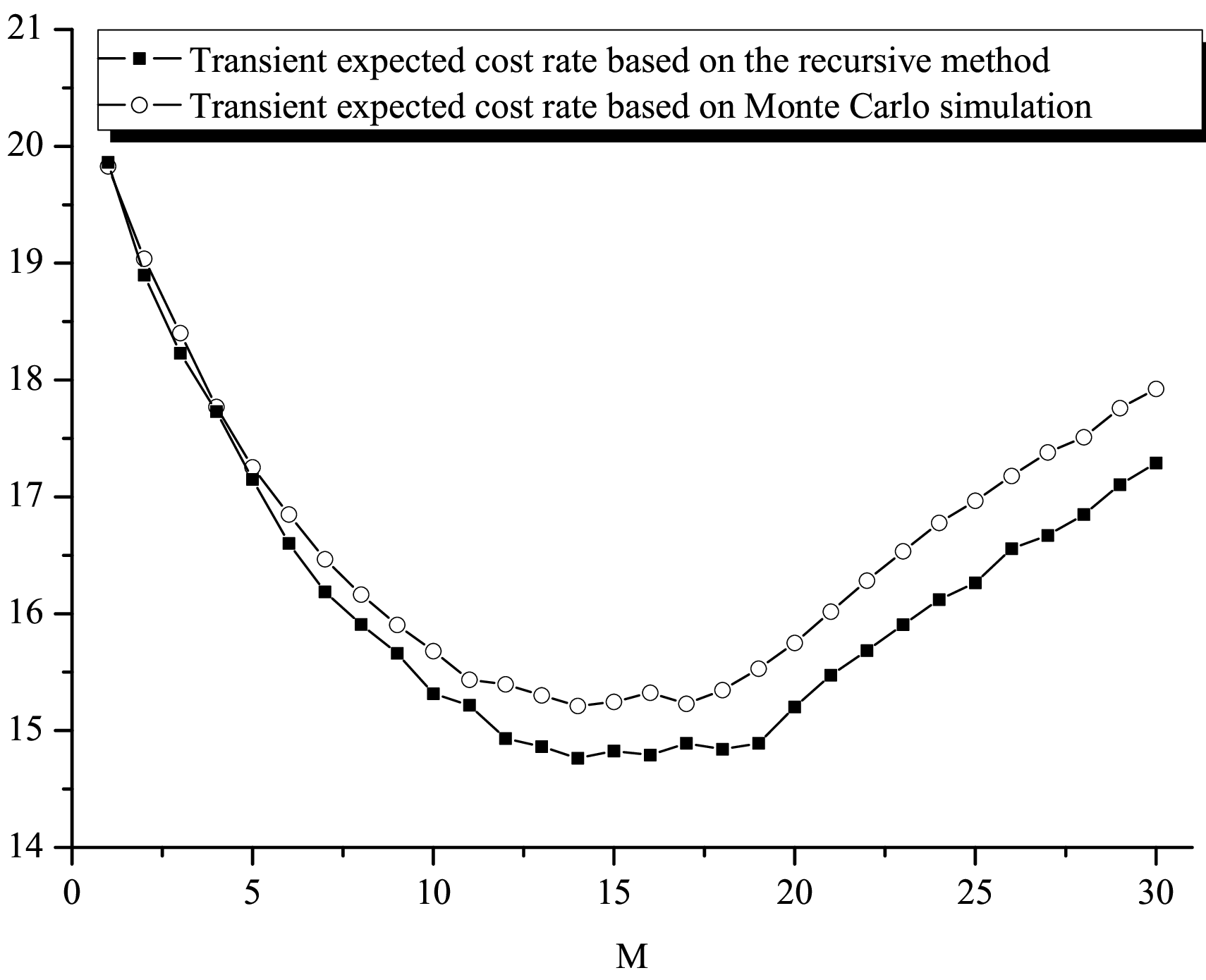}
\caption{Expected cost rate in the life cycle for different values of $M$.}\label{costes_generalM}
\end{figure}

For $T=10$ $t.u.$, Fig. \ref{asintoticoM} shows the expected transient cost rate calculated by using the recursive method and the asymptotic expected cost rate calculated throughout the procedure detailed in Section \ref{two-dimensional_aymptotic} versus $M$.
\begin{figure}[h]
\begin{center}
\includegraphics[scale=0.30]{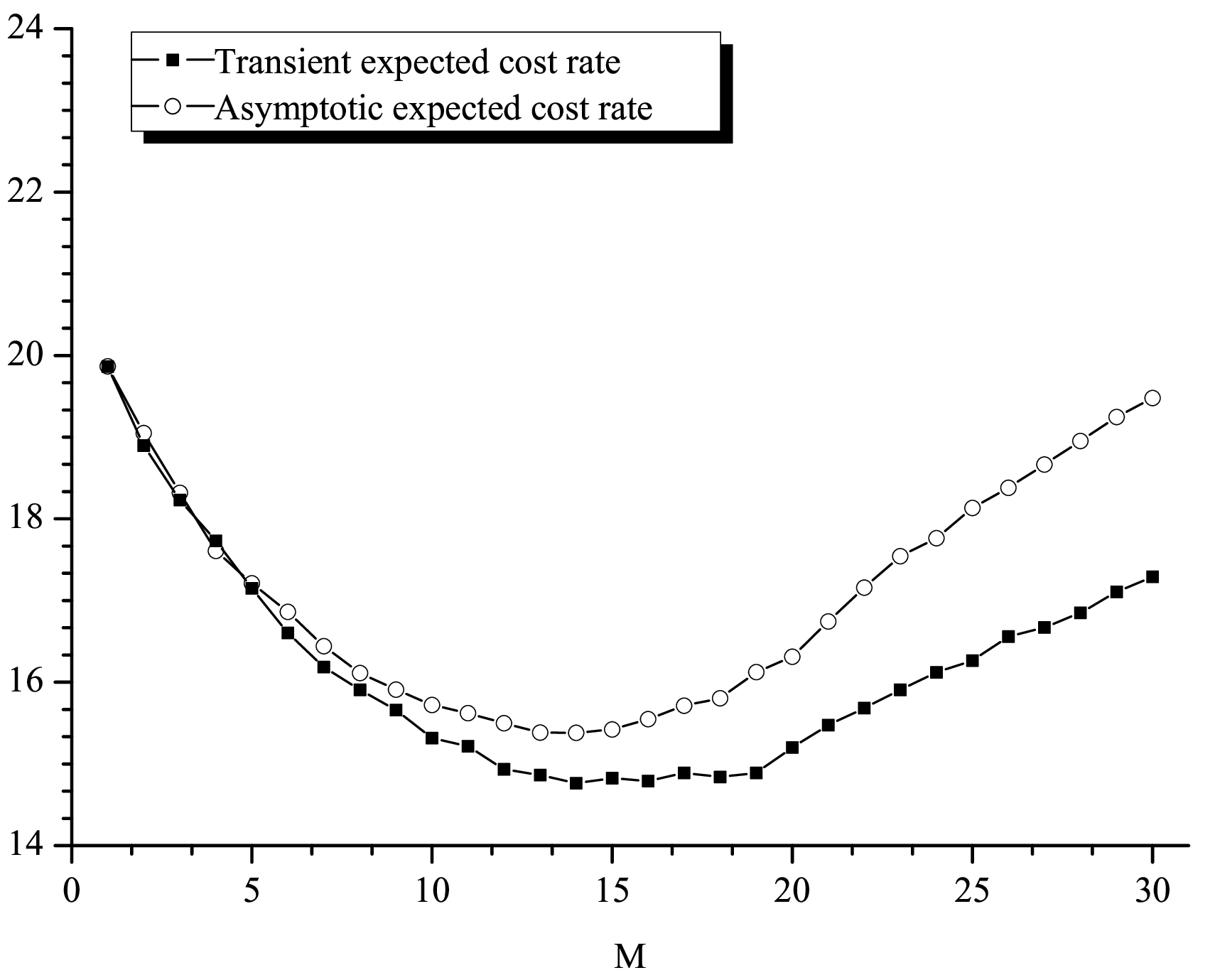}
\caption{Expected transient cost rate and asymptotic expected cost rate at time $t_f=50$ $t.u.$ for different values of $M$.}\label{asintoticoM}
\end{center}
\end{figure}
As we said previously, the value of $M$ which minimises the expected transient cost rate is reached at $M=14$ $d.u.$, with an expected cost rate of $14.7639$ $m.u./t.u.$ On the other hand, the asymptotic expected cost rate reaches its minimum value at $M=14$ $d.u.$, with an expected cost rate of $15.3819$ $m.u./t.u.$

Now, we calculate the standard deviation of the cost. Fig. \ref{sd_M} shows the expected transient cost rate with its standard deviation associated given by Equation (\ref{standard_deviation}). Both quatities were calculated for $30$ equally spaced points in $(0,30]$ by using the recursive formula given in (\ref{Bt}) throughout the steps detailed in \ref{procedimiento_recursivo_variandoM}.
\begin{figure}[h]
\begin{center}
\includegraphics[scale=0.30]{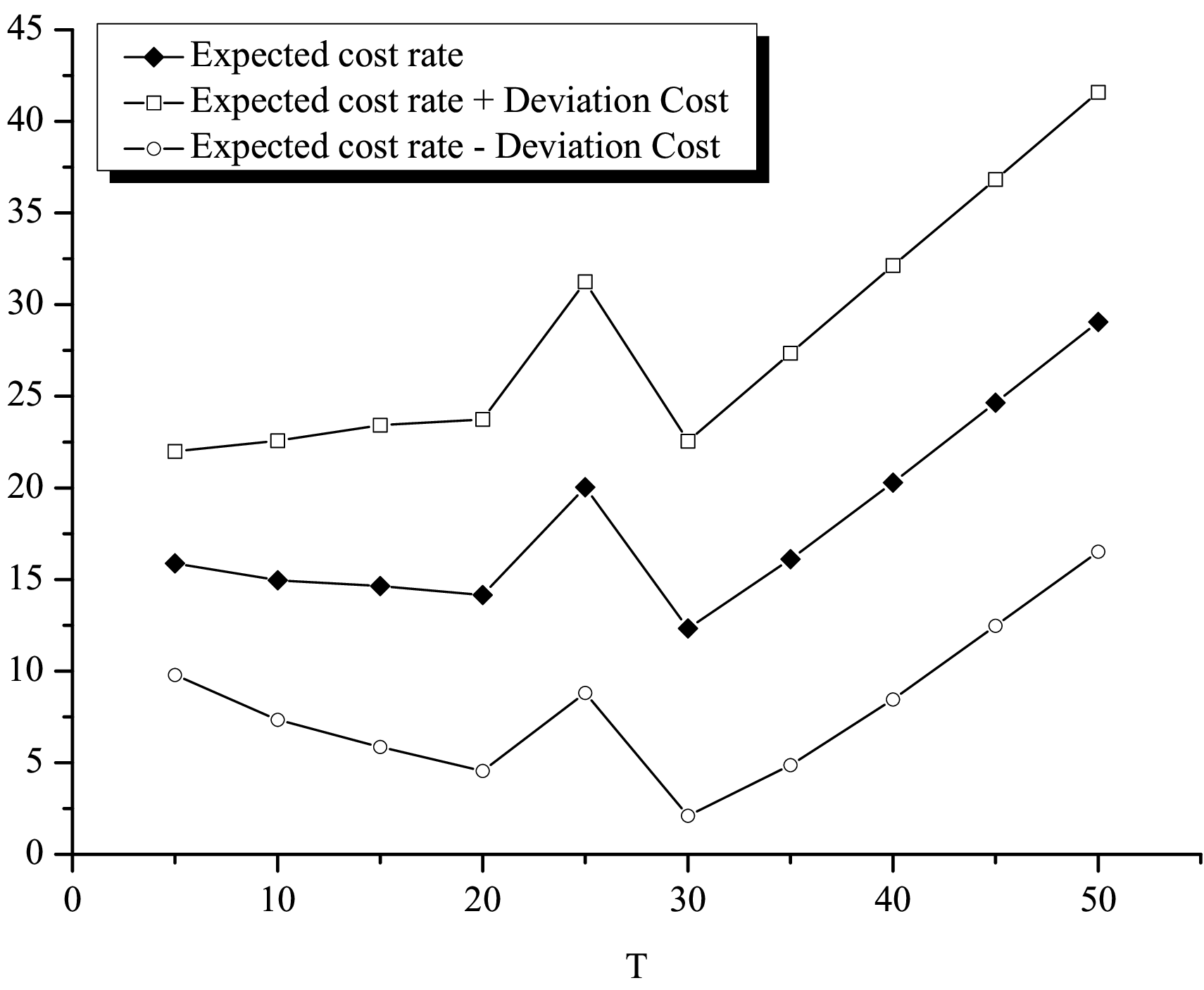}
\caption{Expected cost rate in the life cycle and standard deviation versus $M$.}\label{sd_M}
\end{center}
\end{figure}

Now, we focus on the influence of the main model parameters on the expected cost in the life cycle. Firstly, a sensitivity analysis of the gamma process
parameters is performed.

The values of the gamma process parameters are modified according to
the following specifications:
\begin{equation}\label{especificaciones_alpha_beta}
\alpha_{(v_i\%)}=\alpha\left[1+
\frac{v_i}{100}\right]\quad\text{and}\quad\beta_{(v_j\%)}=\beta\left[1+\frac{v_j}{100}\right],
\end{equation}
where $v_i$ and $v_j$  are, respectively, the $i$-th and $j$-th
position of the vector $\mathbf{v}=(-10$, $-5$, $-1$, $0$, $1$, $5$,
$10$). Then, the parameter values for $\alpha$ and $\beta$ can be
simultaneous and independently modified both for increasing and
decreasing changes.

Let $\tilde{E}\left[C_{10,\alpha_{(v_i\%)},\beta_{(v_j\%)}}^{M}(t_f)\right]$ be the minimal expected transient cost at time $t_f$ obtained when the gamma process parameters ($\alpha$ and $\beta$) are varied according to the specifications given in (\ref{especificaciones_alpha_beta}). The expected transient cost for each combination of $\alpha_{(v_i\%)}$ and $\beta_{(v_j\%)}$ are calculated based on the recursive method following the steps detailed in \ref{procedimiento_recursivo_variandoM}. The relative measure $V_{10,\alpha_{(v_i\%)},\beta_{(v_j\%)}}^{M}(50)$ is defined as
\begin{equation}\label{variacion_alpha_beta}
\frac{\left|\tilde{E}\left[C_{10}^{M_{opt}}(50)\right]-\tilde{E}\left[C_{10,\alpha_{(v_i\%)},\beta_{(v_j\%)}}^{M}(50)\right]\right|}{\tilde{E}\left[C_{10}^{M_{opt}}(50)\right]},
\end{equation}
where $\tilde{E}\left[C_{10}^{M_{opt}}(50)\right]$ is the minimal expected transient cost calculated in \ref{procedimiento_recursivo_variandoM}.

For fixed $i$ and $j$, $V_{10,\alpha_{(v_i\%)},\beta_{(v_j\%)}}^{M}(50)$
measures the relative difference between the minimal expected transient cost with the original parameter values
and the minimal expected transient cost calculated using the modified
parameter values for fixed $M$ and $T=10$ $t.u.$ Values closer to zero have a lower influence
on the expected transient cost rate.

Table \ref{variation_alpha_beta_Tfixed} shows the relative variation percentages with a shaded grey
scale. Each cell represents $V_{10,\alpha_{(v_i\%)},\beta_{(v_j\%)}}^{M}(50)$ expressed in percentage. Darker colours of cells denote a higher relative
variation percentage. The results obtained show that $V_{10,\alpha_{(v_i\%)},\beta_{(v_j\%)}}^{M}(50)$ grows when $\alpha$ increases and $\beta$ decreases and $V_{10,\alpha_{(v_i\%)},\beta_{(v_j\%)}}^{M}(50)$ decreases when $\alpha$ decreases and $\beta$ increases. In this way, $V_{10,\alpha_{(v_i\%)},\beta_{(v_j\%)}}^{M}(50)$ reaches its minimum value when $\alpha$ is minimum and $\beta$ is maximum and  its maximum value when $\alpha$ is maximum and $\beta$ is minimum.

\begin{table}
\begin{minipage}[b]{\linewidth}\centering
\renewcommand{\arraystretch}{2}
\scalebox{0.68}{
\begin{tabular}{l|ccccccc|}
\cline{2-8} & $\beta_{_{(-10\%)}}$ & $\beta_{_{(-5\%)}}$ &
$\beta_{_{(-1\%)}}$ & $\beta$ & $\beta_{_{(1\%)}}$ &
$\beta_{_{(5\%)}}$ &
$\beta_{_{(10\%)}}$  \\
\hline

\multicolumn{1}{|l|}{$\alpha_{(-10\%)}$} & \cellcolor[gray]{0.961218} 1.1862 & \cellcolor[gray]{0.925243} 2.2865 & \cellcolor[gray]{0.809380} 5.8303 & \cellcolor[gray]{0.812387} 5.7383 & \cellcolor[gray]{0.789536} 6.4372 & \cellcolor[gray]{0.716511} 8.6707 & \cellcolor[gray]{0.638045} 11.0707 \\

\multicolumn{1}{|l|}{$\alpha_{(-5\%)}$} & \cellcolor[gray]{0.865731} 4.1067 & \cellcolor[gray]{0.968033} 0.9778 & \cellcolor[gray]{0.935804} 1.9635 & \cellcolor[gray]{0.910406} 2.7403 & \cellcolor[gray]{0.895126} 3.2077 & \cellcolor[gray]{0.801052} 6.0850 & \cellcolor[gray]{0.703849} 9.0580 \\

\multicolumn{1}{|l|}{$\alpha_{(-1\%)}$} &  \cellcolor[gray]{0.770864} 7.0083 & \cellcolor[gray]{0.910626} 2.7336 & \cellcolor[gray]{0.979162} 0.6373 & \cellcolor[gray]{0.988013} 0.3666 & \cellcolor[gray]{0.964349} 1.0904 & \cellcolor[gray]{0.881929} 3.6113 & \cellcolor[gray]{0.799320} 6.1380 \\

\multicolumn{1}{|l|}{$\alpha$} & \cellcolor[gray]{0.737828} 8.0187 & \cellcolor[gray]{0.874935} 3.8252 & \cellcolor[gray]{0.965595} 1.0523 & \cellcolor[gray]{1.000000} 0.0000 & \cellcolor[gray]{0.991420} 0.2624 & \cellcolor[gray]{0.908351} 2.8032 & \cellcolor[gray]{0.802757} 6.0328 \\

\multicolumn{1}{|l|}{$\alpha_{(1\%)}$} & \cellcolor[gray]{0.715212} 8.7105 & \cellcolor[gray]{0.850501} 4.5726 & \cellcolor[gray]{0.952568} 1.4507 & \cellcolor[gray]{0.965029} 1.0696 & \cellcolor[gray]{0.999128} 0.0267 & \cellcolor[gray]{0.913896} 2.6336 & \cellcolor[gray]{0.832088} 5.1357 \\

\multicolumn{1}{|l|}{$\alpha_{(5\%)}$} & \cellcolor[gray]{0.630547} 11.3000 & \cellcolor[gray]{0.768347} 7.0853 & \cellcolor[gray]{0.867416} 4.0552 & \cellcolor[gray]{0.880496} 3.6551 & \cellcolor[gray]{0.910915} 2.7247 & \cellcolor[gray]{0.992336} 0.2344 & \cellcolor[gray]{0.909770} 2.7598 \\

\multicolumn{1}{|l|}{$\alpha_{(10\%)}$} & \cellcolor[gray]{0.500000} 15.2929 & \cellcolor[gray]{0.654189} 10.5769 & \cellcolor[gray]{0.762858} 7.2532 & \cellcolor[gray]{0.781977} 6.6684 & \cellcolor[gray]{0.811460} 5.7666 & \cellcolor[gray]{0.899942} 3.0604 & \cellcolor[gray]{0.996866} 0.0958 \\
\hline
\end{tabular}}
\caption{Relative variation percentages for the expected transient cost for the gamma process parameters for a fixed $T=10$ $t.u.$}\label{variation_alpha_beta_Tfixed} \vspace{1cm}
\end{minipage}
\end{table}

By modifying $\pm 1\%$ around $\alpha=\beta=0.1$, the relative variation percentages are small. The
results also show that the relative variation percentages are lower
in the diagonal of the table. That means when the parameters
$\alpha$ and $\beta$ are modified in the same direction and
magnitude. 

Similarly, the values of the parameters $\lambda_1$ and $\lambda_2$ are
modified according to the following specifications:
\begin{equation}\label{especificaciones_lambda1_lambda2}
\lambda_{1,(v_i\%)}=\lambda_1\left[1+
\frac{v_i}{100}\right]\quad\text{and}\quad\lambda_{2,(v_j\%)}=\lambda_2\left[1+\frac{v_j}{100}\right],
\end{equation}

Let $\tilde{E}^*\left[C^M_{10,\lambda_{1,(v_i\%)},\lambda_{2,(v_j\%)}}(t_f)\right]$ be the minimal expected transient cost
obtained by varying the parameters $\lambda_1$ and $\lambda_2$
simultaneously as in the scheme given in (\ref{especificaciones_lambda1_lambda2}). Now, the relative variation $V^M_{10,\lambda_{1,(v_i\%)},\lambda_{2,(v_j\%)}}(50)$
is given by
\begin{equation}\label{variacion_choques}
\frac{\left|\tilde{E}^*\left[C_{10}^{M_{opt}}(50)\right]-\tilde{E}^*\left[C_{10,\lambda_{1,(v_i\%)},\lambda_{2,(v_j\%)}}^{M}(50)\right]\right|}{\tilde{E}^*\left[C_{10}^{M_{opt}}(50)\right]}.
\end{equation}

The relative variation percentages are presented in Table \ref{variation_lambda1_lambda2_Tfixed}. The results show that the parameter $\lambda_1$ has greater effects on $V^M_{10,\lambda_{1,(v_i\%)},\lambda_{2,(v_j\%)}}(50)$
than the parameter $\lambda_2$, reaching the lowest values when the variation for $\lambda_1=0.01$ is minimal, that is $\pm 1\%$, and the highest values when the variation for $\lambda_1$ is maximised, that is $\pm 10\%$.
\begin{table}
\begin{minipage}[b]{\linewidth}\centering
\renewcommand{\arraystretch}{2}
\scalebox{0.68}{
\begin{tabular}{l|ccccccc|}
\cline{2-8} & $\lambda_{2,(-10\%)}$ & $\lambda_{2,(-5\%)}$ & $\lambda_{2,(-1\%)}$ &
$\lambda_{2}$ & $\lambda_{2,(1\%)}$ & $\lambda_{2,(5\%)}$ &
$\lambda_{2,(10\%)}$  \\
\hline

\multicolumn{1}{|l|}{$\lambda_{1,(-10\%)}$} & \cellcolor[gray]{0.530281} 3.0794 & \cellcolor[gray]{0.622304} 2.4761 & \cellcolor[gray]{0.645110} 2.3266 & \cellcolor[gray]{0.704476} 1.9374 & \cellcolor[gray]{0.712491} 1.8849 & \cellcolor[gray]{0.690594} 2.0284 & \cellcolor[gray]{0.790972} 1.3703 \\

\multicolumn{1}{|l|}{$\lambda_{1,(-5\%)}$} & \cellcolor[gray]{0.754491} 1.6095 & \cellcolor[gray]{0.810563} 1.2419 & \cellcolor[gray]{0.792840} 1.3581 & \cellcolor[gray]{0.863547} 0.8946 & \cellcolor[gray]{0.862690} 0.9002 & \cellcolor[gray]{0.953351} 0.3058 & \cellcolor[gray]{0.942717} 0.3755 \\

\multicolumn{1}{|l|}{$\lambda_{1,(-1\%)}$} &  \cellcolor[gray]{0.906057} 0.6159 & \cellcolor[gray]{0.920268} 0.5227 & \cellcolor[gray]{0.965907} 0.2235 & \cellcolor[gray]{1.000000} 0.0000 & \cellcolor[gray]{0.918353} 0.5353 & \cellcolor[gray]{0.996322} 0.0241 & \cellcolor[gray]{0.857156} 0.9365 \\

\multicolumn{1}{|l|}{$\lambda_{1}$} & \cellcolor[gray]{0.929863} 0.4598 & \cellcolor[gray]{0.955764} 0.2900 & \cellcolor[gray]{0.965907} 0.2235 & \cellcolor[gray]{1.000000} 0.0000 & \cellcolor[gray]{0.918353} 0.5353 & \cellcolor[gray]{0.996322} 0.0241 & \cellcolor[gray]{0.857156} 0.9365 \\

\multicolumn{1}{|l|}{$\lambda_{1,(1\%)}$} & \cellcolor[gray]{0.929863} 0.4598 & \cellcolor[gray]{0.962607} 0.2451 & \cellcolor[gray]{0.924558} 0.4946 & \cellcolor[gray]{0.998200} 0.0118 & \cellcolor[gray]{0.949824} 0.3289 & \cellcolor[gray]{0.839096} 1.0549 & \cellcolor[gray]{0.850018} 0.9833 \\

\multicolumn{1}{|l|}{$\lambda_{1,(5\%)}$} & \cellcolor[gray]{0.958839} 0.2698 & \cellcolor[gray]{0.943787} 0.3685 & \cellcolor[gray]{0.910834} 0.5846 & \cellcolor[gray]{0.888326} 0.7321 & \cellcolor[gray]{0.928035} 0.4718 & \cellcolor[gray]{0.872406} 0.8365 & \cellcolor[gray]{0.823762} 1.1554 \\

\multicolumn{1}{|l|}{$\lambda_{1,(10\%)}$} & \cellcolor[gray]{0.795894} 1.3381 & \cellcolor[gray]{0.642988} 2.3405 & \cellcolor[gray]{0.587459} 2.7045 & \cellcolor[gray]{0.656213} 2.2538 & \cellcolor[gray]{0.573455} 2.7964 & \cellcolor[gray]{0.609608} 2.5593 & \cellcolor[gray]{0.500000} 3.2779 \\
\hline
\end{tabular}}
\caption{Relative variation percentages for the expected transient cost for parameters $\lambda_1$ and $\lambda_2$ for a fixed $T=10$ $t.u.$}\label{variation_lambda1_lambda2_Tfixed}
\end{minipage}
\end{table}

\subsection{Analysis of the expected cost rate in the life cycle for a fixed $M$}\label{Section_M_fixed}
We now analyse the time between inspections $T$ influence on the expected cost in the life cycle of the system for a preventive threshold $M=14$ $d.u.$ As in Section \ref{procedimiento_recursivo_variandoM}, the optimisation problem for the expected transient cost based on the recursive formula given in (\ref{ECt})  is computed throughout the following steps:

\begin{enumerate} \label{procedimiento_recursivo_variandoT}
\item A grid of size $10$ is obtained by discretising the set $[5,50]$ into $10$ equally spaced points from $5$ to $50$ for the time between inspections $T$. 

\item For fixed $M=14$ $d.u.$ and for each fixed $T_i$, we obtain $50000$ simulations of $(R_1,I_1,W_d)$. With these simulations and applying Monte Carlo method, we obtain the estimations $\tilde{P}_{R_{1,p}}^{14}(kT_i)$, $\tilde{P}_{R_{1,c}}^{14}(kT_i)$, $\tilde{P}_{R_1}^{14}(kT_i)$, and $\tilde{E}\left[W_{T_i}^{14}((k-1)T_i,kT_i)\right]$ for $k=1,2,\ldots,\lfloor 50/T_i\rfloor$ (see Fig. \ref{flujos_variando_T_M}).

\item Let $\tilde{E}\left[C_{T_i}^{14}(50)\right]$ be the expected cost in the life cycle. The expected cost is calculated by using the recursive formula given in (\ref{ECt}), replacing the corresponding probabilities by their estimations calculated in Step 2, with initial condition $\tilde{E}\left[C_{T_i}^{14}(0)\right]=0$. 

\item For fixed $M=14$ $d.u.$, the optimisation problem is reduced to find the value of $T_{opt}$ which minimises $\tilde{E}\left[C_{T}^{14}(50)\right]$, that is,
\begin{equation*}
\tilde{E}\left[C_{T_{opt}}^{14}(50)\right]=\min_{T\geq 0}\left\{\tilde{E}\left[C_{T}^{14}(50)\right]\right\}.
\end{equation*} 

\end{enumerate}

Let $\tilde{E}\left[C_{T_i}^{14}(50)\right]/50$ be the expected cost rate in the life cycle of the system. The expected cost rate evaluated using the recursive formula and using Monte Carlo simulation are shown in Fig. \ref{costes_generalT}. The expected cost rate based on strictly Monte Carlo simulation was calculated for $10$ equally spaced points in the interval $(0,50]$ with $50000$ realizations for each point. 
\begin{figure}[h]
\begin{center}
\includegraphics[scale=0.30]{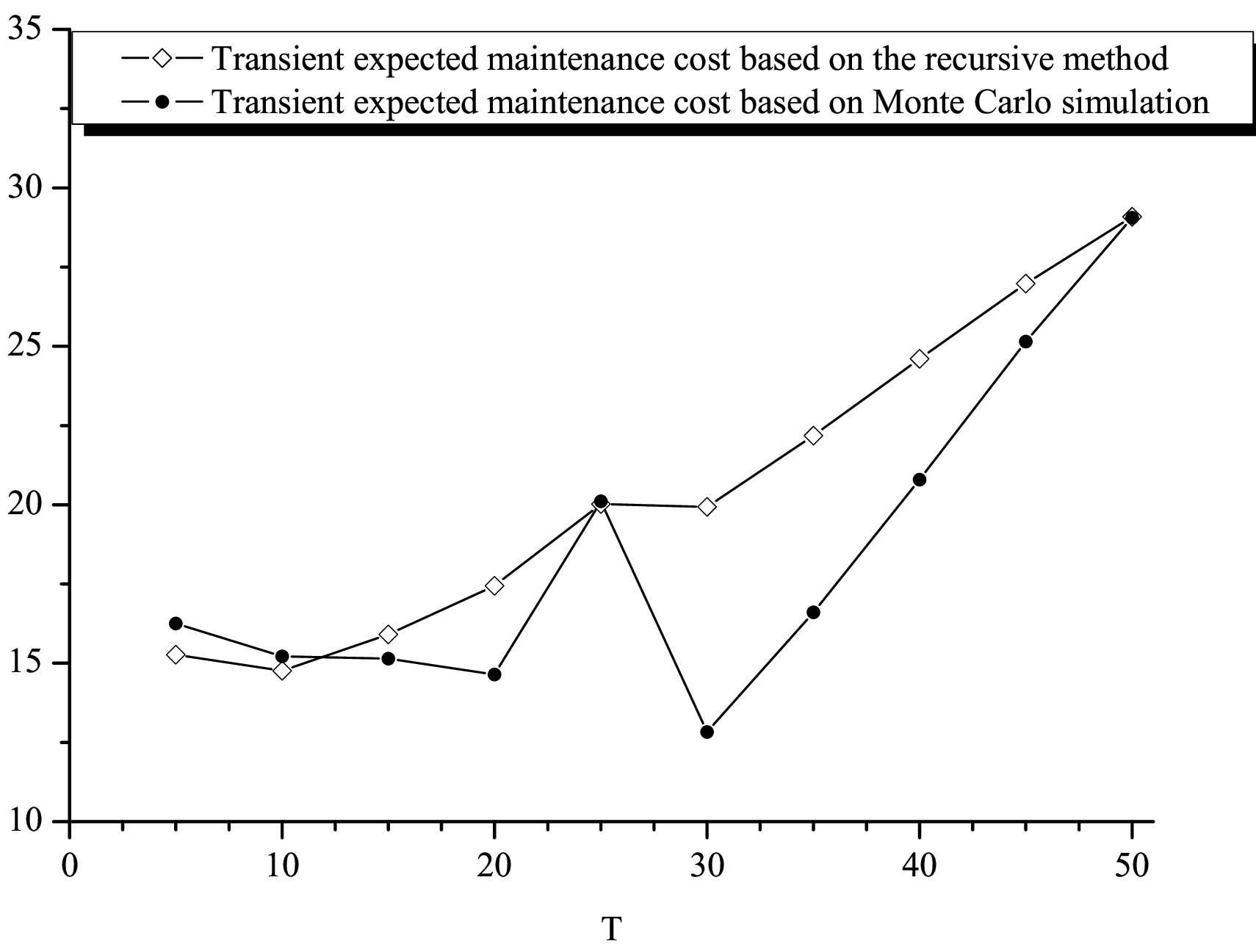}
\caption{Expected transient cost rate at time $t_f=50$ $t.u.$ versus $T$.}\label{costes_generalT}
\end{center}
\end{figure}
Based on Fig. \ref{costes_generalT}, for the recursive method, the expected cost rate at time $t_f=50$ $t.u.$  reaches its minimum value for $T=10$ $t.u.$, with an expected cost rate of $14.7637$ $m.u./t.u.$ On the other hand, using strictly Monte Carlo simulation, the expected cost rate at time $t_f=50$ $t.u.$ reaches its minimum value at $T=30$ $t.u.$, with an expected cost rate of $12.8252$ $m.u./t.u.$ This difference can be explained due to the high variance in deterioration increments of the degradation process.

For $M=14$ $d.u.$ Table \ref{numciclos_promedioT} shows the average number of renewals of the system in its life cycle for each value of $T$.

\begin{table}[h]
  \centering
  \scalebox{0.75}{
    \begin{tabular}{|c|c|c|c|c|c|}
    \hline
    $\mathbf{T}$ & \textbf{5}     & \textbf{10}     & \textbf{15}     & \textbf{20}    & \textbf{25} \\
    $\mathbf{E\left[N_{T}^{14}(50)\right]}$ & 2.4650 & 2.2007 & 1.7772 & 1.4327 & 1.6419 \\
    \hline
    \multicolumn{6}{c}{}\\
    \hline
    $\mathbf{T}$ & \textbf{30}    & \textbf{35}    & \textbf{40}    & \textbf{45}    & \textbf{50} \\

     $\mathbf{E\left[N_{T}^{14}(50)\right]}$ & 0.8884 & 0.93964 & 0.9682 & 0.9833 & 0.9926
 \\
    \hline
    \end{tabular}} 
    \caption{Average number of complete renewal cycles up to $t_f=50$ for different values of $T$.}
  \label{numciclos_promedioT}
\end{table}

For fixed $M=14$ $d.u.$, Fig. \ref{asintoticoT} shows the expected cost rate in the life cycle of the system calculated by using the recursive method and the asymptotic expected cost rate calculated throughout the procedure detailed in Section \ref{two-dimensional_aymptotic} versus $T$.
\begin{figure}[h]
\begin{center}
\includegraphics[scale=0.28]{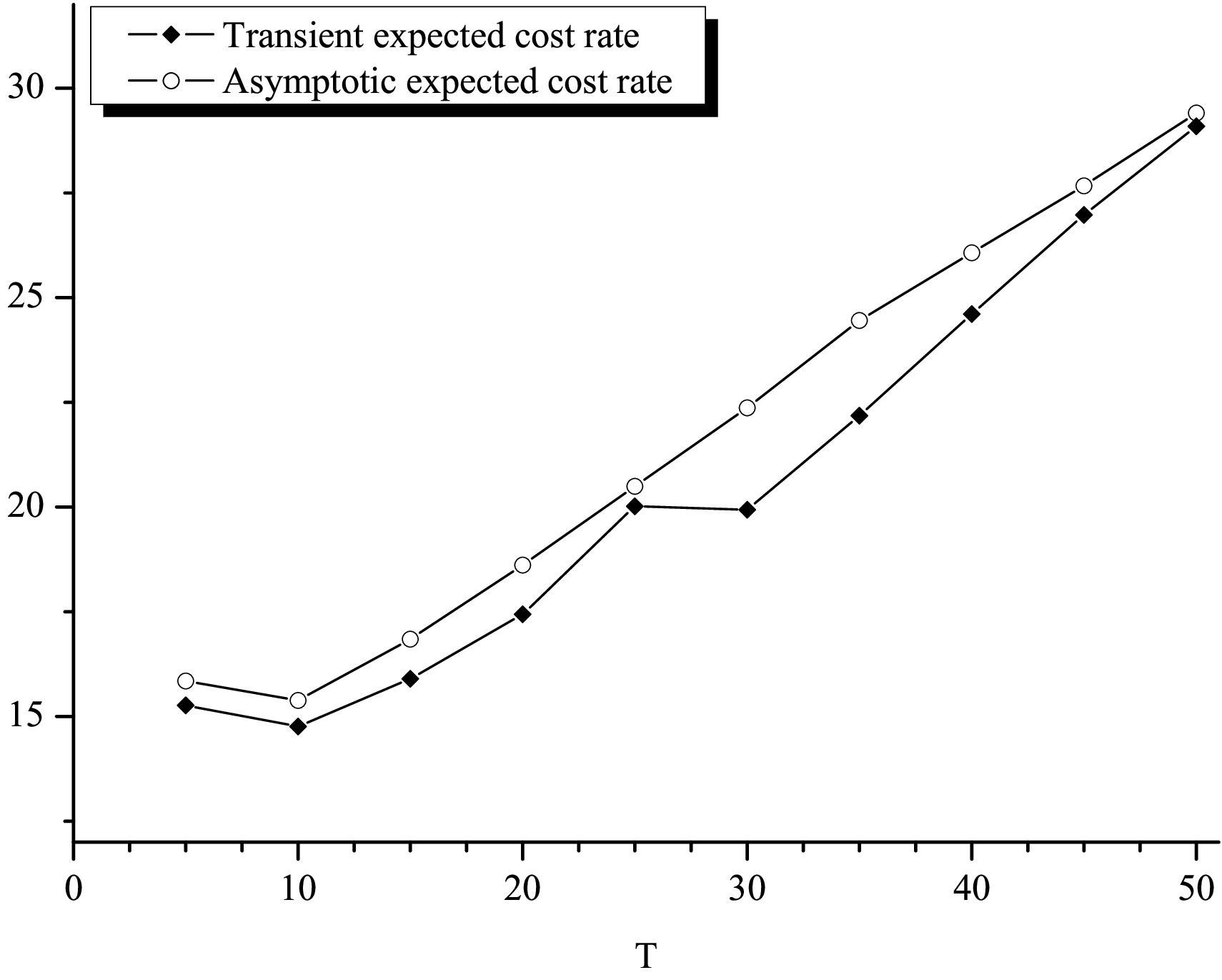}
\caption{Expected transient cost rate and asymptotic expected cost rate at time $t_f=50$ $t.u.$ for different values of $T$.}\label{asintoticoT}
\end{center}
\end{figure}
As we said previously, the expected transient cost rate calculated using the recursive method at time $t_f=50$ $t.u.$ reaches its minimum value at $T=10$ $t.u.$, with an expected transient cost rate of $14.7637$ $m.u./t.u.$ On the other hand, the asymptotic expected cost rate reaches its minimum value at $T=10$ $t.u.$, with an asymptotic expected cost rate of $15.3819$ $m.u./t.u.$ Asymptotic expected cost rate shows a smoother behaviour compared to the expected transient cost rate.

Next, the standard deviation of the cost is calculated. Fig. \ref{sd_T} shows the expected transient cost rate calculated using the recursive method with its standard deviation associated given in (\ref{standard_deviation}) versus $T$. This deviation was calculated for $10$ equally spaced points in the interval $(0,50]$ by using the recursive formula given in Equation (\ref{Bt}) throughout the steps detailed in \ref{procedimiento_recursivo_variandoT}.
\begin{figure}[h]
\begin{center}
\includegraphics[scale=0.30]{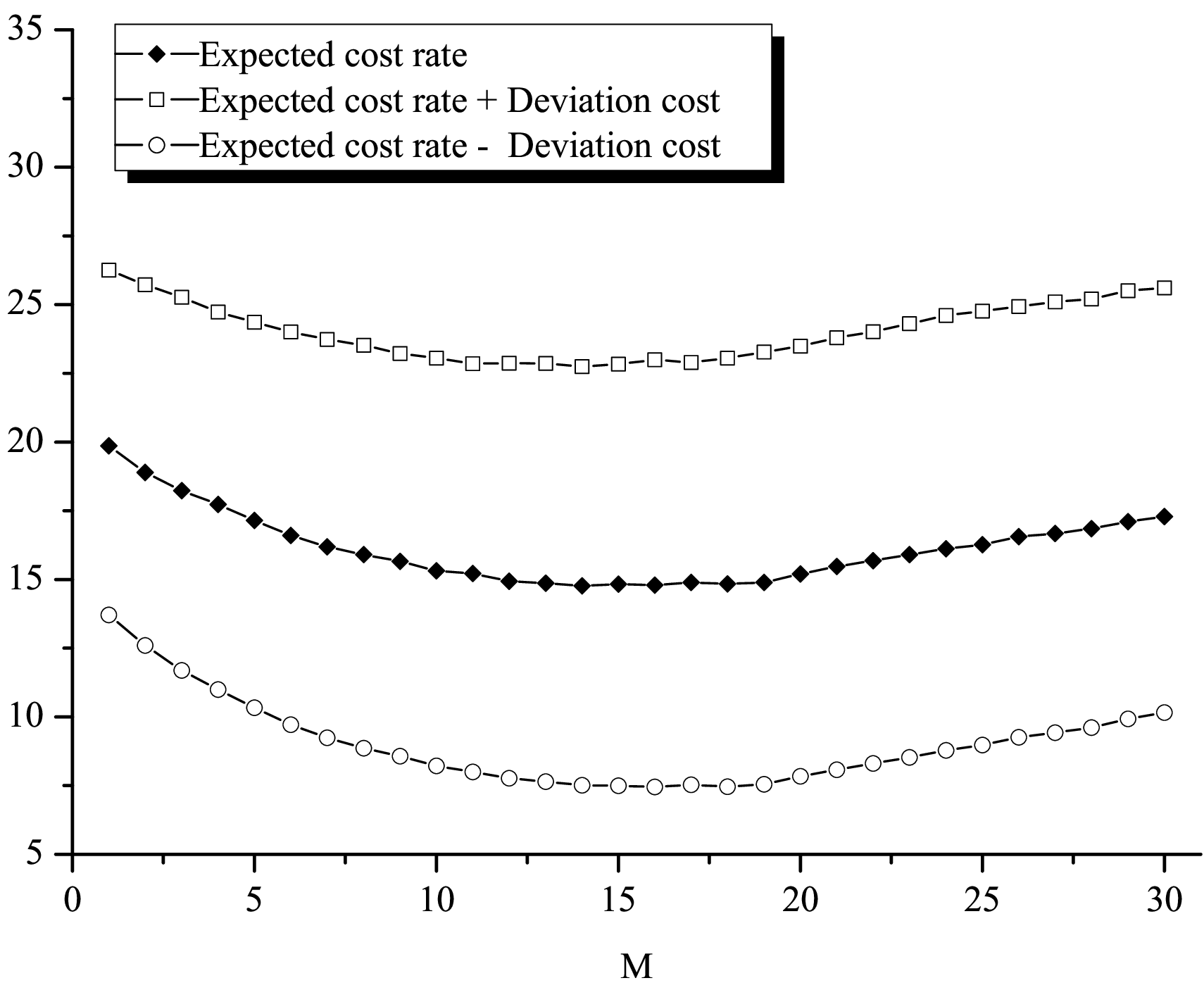}
\caption{Expected transient cost rate at time $t_f=50$ $t.u.$ and its standard deviation associated for different values of $T$.}\label{sd_T}
\end{center}
\end{figure}

Focusing now on the main model parameters influence on the
solution, we analyse the gamma process
parameters sensitivity.

Let $\tilde{E}\left[C_{T,\alpha_{(v_i\%)},\beta_{(v_j\%)}}^{14}(50)\right]$ be the minimal expected transient cost at time $t_f=50$ $t.u.$ obtained when the gamma process parameters ($\alpha$ and $\beta$) are varied according to the specifications given in (\ref{especificaciones_alpha_beta}) for each value of $T$ with preventive threshold $M=14$ $d.u.$ Based on (\ref{variacion_alpha_beta}), $V_{T,\alpha_{(v_i\%)},\beta_{(v_j\%)}}^{14}(50)$ denotes the relative variation between the minimal expected transient cost with the original parameter values and the minimal expected transient cost calculated by using the parameter values modified  according to (\ref{especificaciones_alpha_beta}) for variable $T$ and $M=14$ $d.u.$

Table \ref{variation_alpha_beta_Tfixed} shows the values obtained for $V_{T,\alpha_{(v_i\%)},\beta_{(v_j\%)}}^{14}(50)$ expressed in percentage. By modifying $\pm 1\%$ around $\alpha=\beta=0.1$, the relative variation percentages are small. The results obtained also show that $V_{T,\alpha_{(v_i\%)},\beta_{(v_j\%)}}^{14}(50)$ is lower in the diagonal of the table. That means when the parameters $\alpha$ and $\beta$ are modified in the same direction and magnitude.
 
\begin{table}
\begin{minipage}[b]{\linewidth}\centering
\renewcommand{\arraystretch}{2}
\scalebox{0.68}{
\begin{tabular}{l|ccccccc|}
\cline{2-8} & $\beta_{_{(-10\%)}}$ & $\beta_{_{(-5\%)}}$ &
$\beta_{_{(-1\%)}}$ & $\beta$ & $\beta_{_{(1\%)}}$ &
$\beta_{_{(5\%)}}$ &
$\beta_{_{(10\%)}}$  \\
\hline

\multicolumn{1}{|l|}{$\alpha_{(-10\%)}$} & \cellcolor[gray]{0.948279} 1.5684 & \cellcolor[gray]{0.937800} 1.8862 & \cellcolor[gray]{0.844780} 4.7069 & \cellcolor[gray]{0.811677} 5.7107 & \cellcolor[gray]{0.797355} 6.1450 & \cellcolor[gray]{0.734174} 8.0609 & \cellcolor[gray]{0.629864} 11.2240 \\

\multicolumn{1}{|l|}{$\alpha_{(-5\%)}$} & \cellcolor[gray]{0.848585} 4.5915 & \cellcolor[gray]{0.969000} 0.9400 & \cellcolor[gray]{0.943165} 1.7235 & \cellcolor[gray]{0.916133} 2.5432 & \cellcolor[gray]{0.888336} 3.3861 & \cellcolor[gray]{0.830628} 5.1360 & \cellcolor[gray]{0.731897} 8.1300 \\

\multicolumn{1}{|l|}{$\alpha_{(-1\%)}$} & \cellcolor[gray]{0.719070} 8.5189 & \cellcolor[gray]{0.885840} 3.4618 & \cellcolor[gray]{0.970310} 0.9003 & \cellcolor[gray]{0.998515} 0.0450 & \cellcolor[gray]{0.976937} 0.6994 & \cellcolor[gray]{0.901065} 3.0001 & \cellcolor[gray]{0.794693} 6.2257 \\

\multicolumn{1}{|l|}{$\alpha_{1}$} & \cellcolor[gray]{0.851810} 4.4937 & \cellcolor[gray]{0.726892} 8.2817 & \cellcolor[gray]{0.954677} 1.3744 & \cellcolor[gray]{1.000000} 0.0000 & \cellcolor[gray]{0.994231} 0.1749 & \cellcolor[gray]{0.908346} 2.7793 & \cellcolor[gray]{0.828127} 5.2119 \\

\multicolumn{1}{|l|}{$\alpha_{(1\%)}$} & \cellcolor[gray]{0.704632} 8.9568 & \cellcolor[gray]{0.843989} 4.7309 & \cellcolor[gray]{0.934388} 1.9896 & \cellcolor[gray]{0.939076} 1.8475 & \cellcolor[gray]{0.983724} 0.4936 & \cellcolor[gray]{0.950769} 1.4929 & \cellcolor[gray]{0.833054} 5.0625 \\

\multicolumn{1}{|l|}{$\alpha_{(5\%)}$} & \cellcolor[gray]{0.620618} 11.5044 & \cellcolor[gray]{0.761288} 7.2387 & \cellcolor[gray]{0.834954} 5.0048 & \cellcolor[gray]{0.884098} 3.5146 & \cellcolor[gray]{0.889720} 3.3441 & \cellcolor[gray]{0.973319} 0.8091 & \cellcolor[gray]{0.910635} 2.7099 \\

\multicolumn{1}{|l|}{$\alpha_{(10\%)}$} & \cellcolor[gray]{0.500000} 15.1620 & \cellcolor[gray]{0.646986} 10.7048 & \cellcolor[gray]{0.750153} 7.5764 & \cellcolor[gray]{0.783903} 6.5529 & \cellcolor[gray]{0.776279} 6.7841 & \cellcolor[gray]{0.878556} 3.6827 & \cellcolor[gray]{0.979853} 0.6109 \\
\hline
\end{tabular}}
\caption{Relative variation percentages for the expected transient cost for the gamma process parameters for a fixed $M=14$ $d.u.$}\label{variation_alpha_beta_Tfixed} \vspace{1cm}
\end{minipage}
\end{table}

On the other hand, let $\tilde{E}^*\left[C^{14}_{T,\lambda_{1,(v_i\%)},\lambda_{2,(v_j\%)}}(t_f)\right]$ be the minimal expected transient cost
obtained by varying the parameters $\lambda_1$ and $\lambda_2$
simultaneously as in the scheme given in (\ref{especificaciones_lambda1_lambda2}).  Based on (\ref{variacion_choques}), $V^{14}_{T,\lambda_{1,(v_i\%)},\lambda_{2,(v_j\%)}}(50)$ denotes the relative variation between the minimal expected transient cost  with the original parameter values and the minimal expected transient cost calculated by using the parameter values modified  according to (\ref{especificaciones_lambda1_lambda2}) for variable $T$ and $M=14$ $d.u.$

The relative variation percentages are presented in Table \ref{variation_lambda1_lambda2_Tfixed}. The results show that when $\lambda_1=0.01$ is modified between $-5\%$ and $1\%$, the relative
variation percentages are small. In addition, the
parameter $\lambda_1$ has greater effects on $V^{14}_{T,\lambda_{1,(v_i\%)},\lambda_{2,(v_j\%)}}(50)$
than the parameter $\lambda_2$.
\begin{table}
\begin{minipage}[b]{\linewidth}\centering
\renewcommand{\arraystretch}{2}
\scalebox{0.68}{
\begin{tabular}{l|ccccccc|}
\cline{2-8} & $\lambda_{2,(-10\%)}$ & $\lambda_{2,(-5\%)}$ & $\lambda_{2,(-1\%)}$ &
$\lambda_{2}$ & $\lambda_{2,(1\%)}$ & $\lambda_{2,(5\%)}$ &
$\lambda_{2,(10\%)}$  \\
\hline

\multicolumn{1}{|l|}{$\lambda_{1,(-10\%)}$} & \cellcolor[gray]{0.680240} 2.1506 & \cellcolor[gray]{0.704223} 1.9893 & \cellcolor[gray]{0.713873} 1.9244 & \cellcolor[gray]{0.813186} 1.2564 & \cellcolor[gray]{0.741655} 1.7375 & \cellcolor[gray]{0.780880} 1.4737 & \cellcolor[gray]{0.746594} 1.7043 \\

\multicolumn{1}{|l|}{$\lambda_{1,(-5\%)}$} & \cellcolor[gray]{0.873244} 0.8525 & \cellcolor[gray]{0.941300} 0.3948 & \cellcolor[gray]{0.893907} 0.7135 & \cellcolor[gray]{0.957462} 0.2861 & \cellcolor[gray]{0.996159} 0.0258 & \cellcolor[gray]{0.967311} 0.2199 & \cellcolor[gray]{0.959363} 0.2733 \\

\multicolumn{1}{|l|}{$\lambda_{1,(-1\%)}$} & \cellcolor[gray]{0.873471} 0.8510 & \cellcolor[gray]{0.925075} 0.5039 & \cellcolor[gray]{0.933952} 0.4442 & \cellcolor[gray]{0.906937} 0.6259 & \cellcolor[gray]{0.928798} 0.4789 & \cellcolor[gray]{0.882602} 0.7896 & \cellcolor[gray]{0.800279} 1.3432 \\

\multicolumn{1}{|l|}{$\lambda_{1}$} & \cellcolor[gray]{0.925861} 0.4986 & \cellcolor[gray]{0.925661} 0.5000 & \cellcolor[gray]{0.930667} 0.4663 & \cellcolor[gray]{1.000000} 0.0000 & \cellcolor[gray]{0.888249} 0.7516 & \cellcolor[gray]{0.844257} 1.0475 & \cellcolor[gray]{0.832893} 1.1239 \\

\multicolumn{1}{|l|}{$\lambda_{1,(1\%)}$} & \cellcolor[gray]{0.997122} 0.0194 & \cellcolor[gray]{0.910251} 0.6036 & \cellcolor[gray]{0.851957} 0.9957 & \cellcolor[gray]{0.861597} 0.9308 & \cellcolor[gray]{0.884622} 0.7760 & \cellcolor[gray]{0.834545} 1.1128 & \cellcolor[gray]{0.705607} 1.9800 \\

\multicolumn{1}{|l|}{$\lambda_{1,(5\%)}$} & \cellcolor[gray]{0.716822} 1.9045 & \cellcolor[gray]{0.764047} 1.5869 & \cellcolor[gray]{0.571360} 2.8829 & \cellcolor[gray]{0.778322} 1.4909 & \cellcolor[gray]{0.682619} 2.1346 & \cellcolor[gray]{0.617928} 2.5697 & \cellcolor[gray]{0.553443} 3.0034 \\

\multicolumn{1}{|l|}{$\lambda_{1,(10\%)}$} &  \cellcolor[gray]{0.639178} 2.4267 & \cellcolor[gray]{0.681883} 2.1395 & \cellcolor[gray]{0.583024} 2.8044 & \cellcolor[gray]{0.555961} 2.9864 & \cellcolor[gray]{0.574677} 2.8605 & \cellcolor[gray]{0.500000} 3.3628 & \cellcolor[gray]{0.552647} 3.0087 \\
\hline
\end{tabular}}
\caption{Relative variation percentages for the expected transient cost rate for parameters $\lambda_1$ and $\lambda_2$ for a fixed $M=14$ $d.u.$}\label{variation_lambda1_lambda2_Tfixed}
\end{minipage}
\end{table}

\subsection{Transient two-dimensional expected cost rate analysis}\label{two-dimensional_transient}
The expected transient cost based on the recursive formula given in (\ref{ECt}) considering $M$ and $T$ variables is analysed. The optimisation problem is computed as follows:
\begin{enumerate}
\item A grid of size $10$ is obtained by discretising the set $[5,50]$ into $10$ equally spaced points from $5$ to $50$ for $T$. Let $T_i$ be the $i$-th value of the grid obtained previously, for $i=1,2,\ldots,10$.

\item A grid of size $30$ is obtained by discretising the set $[1,30]$ into $30$ equally spaced points from $1$ to $30$ for $M$. Let $M_j$ be the $j$-th value of $M$ which corresponds to the $i$-th value of the grid obtained previously, for $j=1,2,\ldots,30$.

\item For each fixed combination $(T_i,M_j)$, we obtain $50000$ simulations of $(R_1,I_1,W_d)$. With these simulations and applying Monte Carlo method, $\tilde{P}_{R_{1,p}}^{M_j}(kT_i)$, $\tilde{P}_{R_{1,c}}^{M_j}(kT_i)$, $\tilde{P}_{R_1}^{M_j}(kT_i)$, and $\tilde{E}\left[W_{T_i}^{M_j}((k-1)T_i,kT_i))\right]$ for $k=1,2,\ldots,\lfloor 50/T_i\rfloor$ are obtained (see Fig. \ref{flujos_variando_T_M}).

\item For fixed $T_i$ and $M_j$, let $\tilde{E}\left[C_{T_i}^{M_j}(50)\right]$ be the expected cost at time $t_f=50$ $t.u.$ Expected cost is calculated by using the recursive formula given in (\ref{ECt}), replacing the corresponding probabilities by their estimations calculated in Step 2, with initial condition $\tilde{E}\left[C_{T_i}^{M_j}(0)\right]=0$. 

\item The optimisation problem is reduced to find the values $T_{opt}$ and $M_{opt}$ which minimise the expected cost $\tilde{E}\left[C_{T_i}^{M_j}(50)\right]$. That is 
\begin{equation*}
\tilde{E}\left[C_{T_{opt}}^{M_{opt}}(50)\right]=\min_{\substack{T\geq 0 \\ 0\leq M\leq L}}\left\{\tilde{E}\left[C_{T}^{M}(50)\right]\right\}.
\end{equation*}

\end{enumerate}

Let $\tilde{E}\left[C_{T_i}^{M_j}(50)\right]/50$ be the expected cost rate at time $t_f=50$ $t.u.$ The expected cost rate versus $T$ and $M$ is shown in Fig. \ref{contour_mesh_transient_cost_rate}. 
\begin{figure*}
\centering
\begin{minipage}[H]{.45\textwidth}
\begin{center}
\includegraphics[scale=0.6]{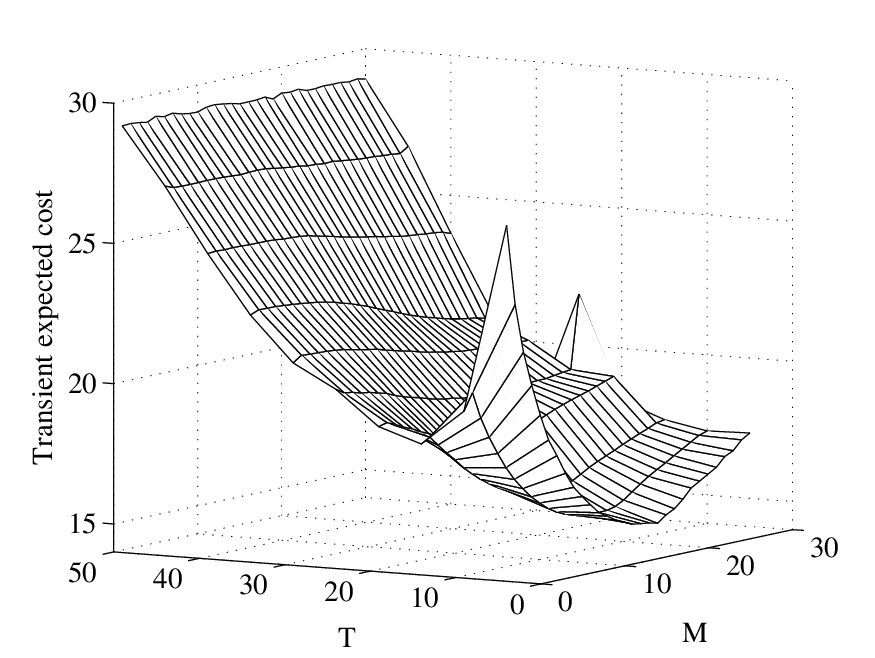}
\end{center}
\end{minipage}
\hfill
\begin{minipage}[H]{.45\textwidth}
\begin{center}
\includegraphics[scale=0.6]{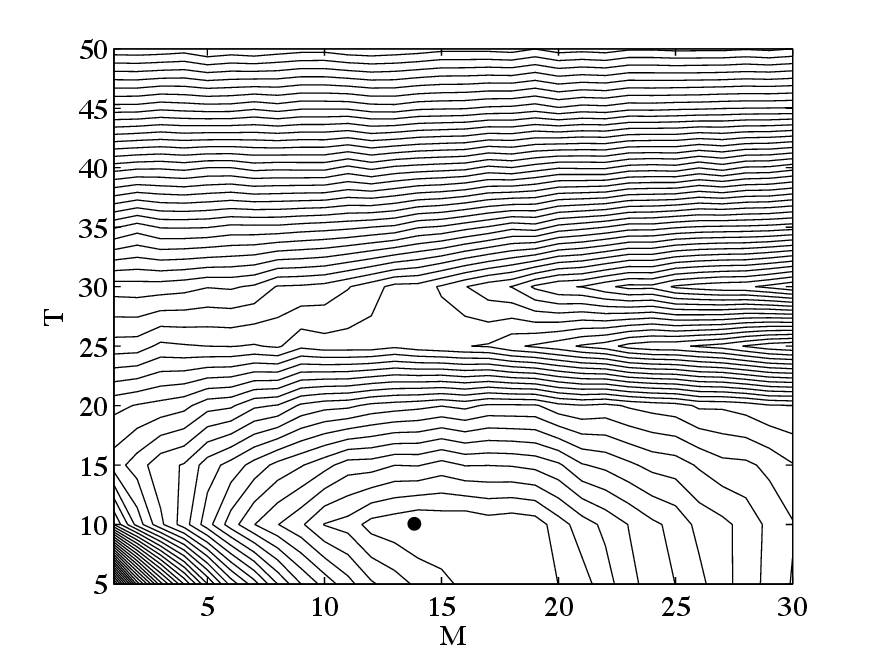}
\end{center}
\end{minipage}
\caption{\label{contour_mesh_transient_cost_rate}Mesh and contour plots for the expected transient cost rate.}
\end{figure*}

The values of $T$ and $M$ which minimise the expected cost rate at time $t_f=50$ $t.u.$ are reached for $M=14$  $d.u.$ and $T=10$ $t.u.$ with an expected cost rate of $14.7637$ $m.u./t.u.$

\subsection{Performance measures for optimal values of $T$ and $M$}
We now analyse the availability and reliability of the system considering the optimal maintenance strategy.

The availability of the system based on the recursive formula given in (\ref{At})  is computed throughout the following steps:

\begin{enumerate} \label{procedimiento_recursivo_variandoT}
\item A grid of size $50$ is obtained by discretising the set $[1,50]$ into $50$ equally spaced points from $1$ to $50$ for the instant time $t$. Let $t_n$ be the $n$-th value of the grid obtained previously, for $n=1,2,\ldots,50$.

\item Let $\tilde{A}_{10}^{14}(t_n)$ be the system availability estimation for a time between inspections $T_{opt}=10$ $t.u.$ and a preventive threshold $M_{opt}=14$ $d.u.$ The availability of the system is calculated by using the recursive formula given in (\ref{At}), replacing $P_{R_1}^{14}(10k)$ by its estimation $\tilde{P}_{R_1}^{14}(10k)$ calculated in Step 2 of the procedure detailed in Section \ref{two-dimensional_transient}, with initial condition $\tilde{A}_{10}^{14}(0)=1$. 
\end{enumerate}

Figure \ref{availabilityFig} shows the availability of the system versus $t$. We can conclude that, for fixed $T=10$ and $M=14$, the probability that the system is working at any instant time of its life cycle is, at least, of the $82\%$.
\begin{figure}[h]
\begin{center}
\includegraphics[scale=0.3]{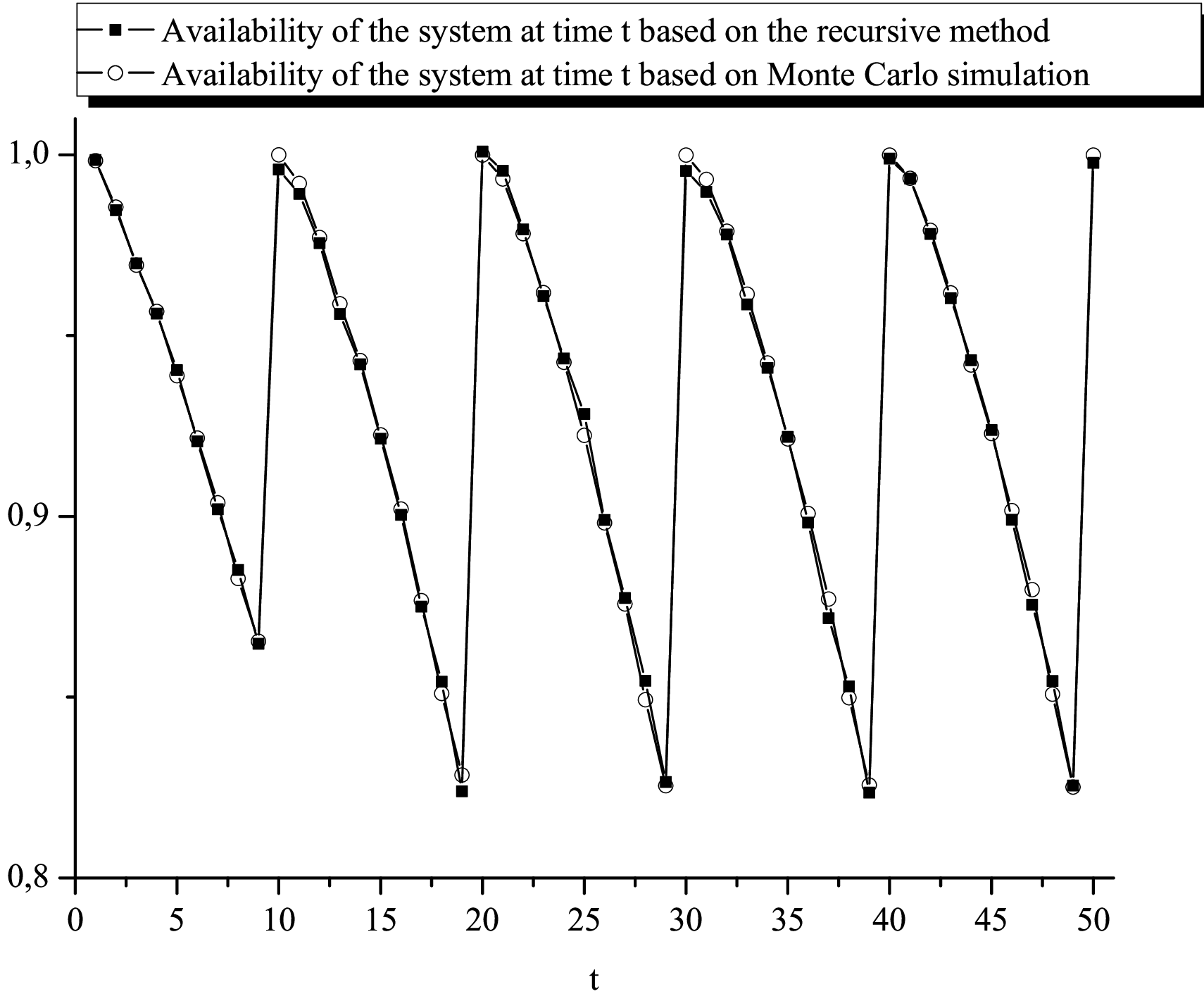}
\caption{Availability of the system for different values of $t$.}\label{availabilityFig}
\end{center}
\end{figure}

Next, the reliability of the system is evaluated. The reliability of the system based on the recursive formula given in (\ref{Rt})  is computed throughout the following steps:

\begin{enumerate} \label{procedimiento_recursivo_variandoT}
\item A grid of size $50$ is obtained by discretising the set $(1,50]$ into $50$ equally spaced points from $1$ to $50$ for $t$. 

\item For fixed $T=10$ and $M=14$, let $\tilde{R}_{10}^{14}(t)$ be the system reliability estimation at time $t$. The reliability of the system is calculated by using the recursive formula given in (\ref{Rt}), replacing $P_{R_{1,p}}^{14}(10k)$ by its estimation $\tilde{P}_{R_{1,p}}^{14}(10k)$ calculated in Step 2 of the procedure detailed in Section \ref{two-dimensional_transient} with initial condition $\tilde{R}_{10}^{14}(0)=1$. 
\end{enumerate}

Figure \ref{reliabilityFig1} shows the reliability of the system versus $t$. We can conclude that, for fixed $T=10$ and $M=14$, the probability that the system does not fail in its life cycle is of the $32\%$.
\begin{figure}[h]
\begin{center}
\includegraphics[scale=0.30]{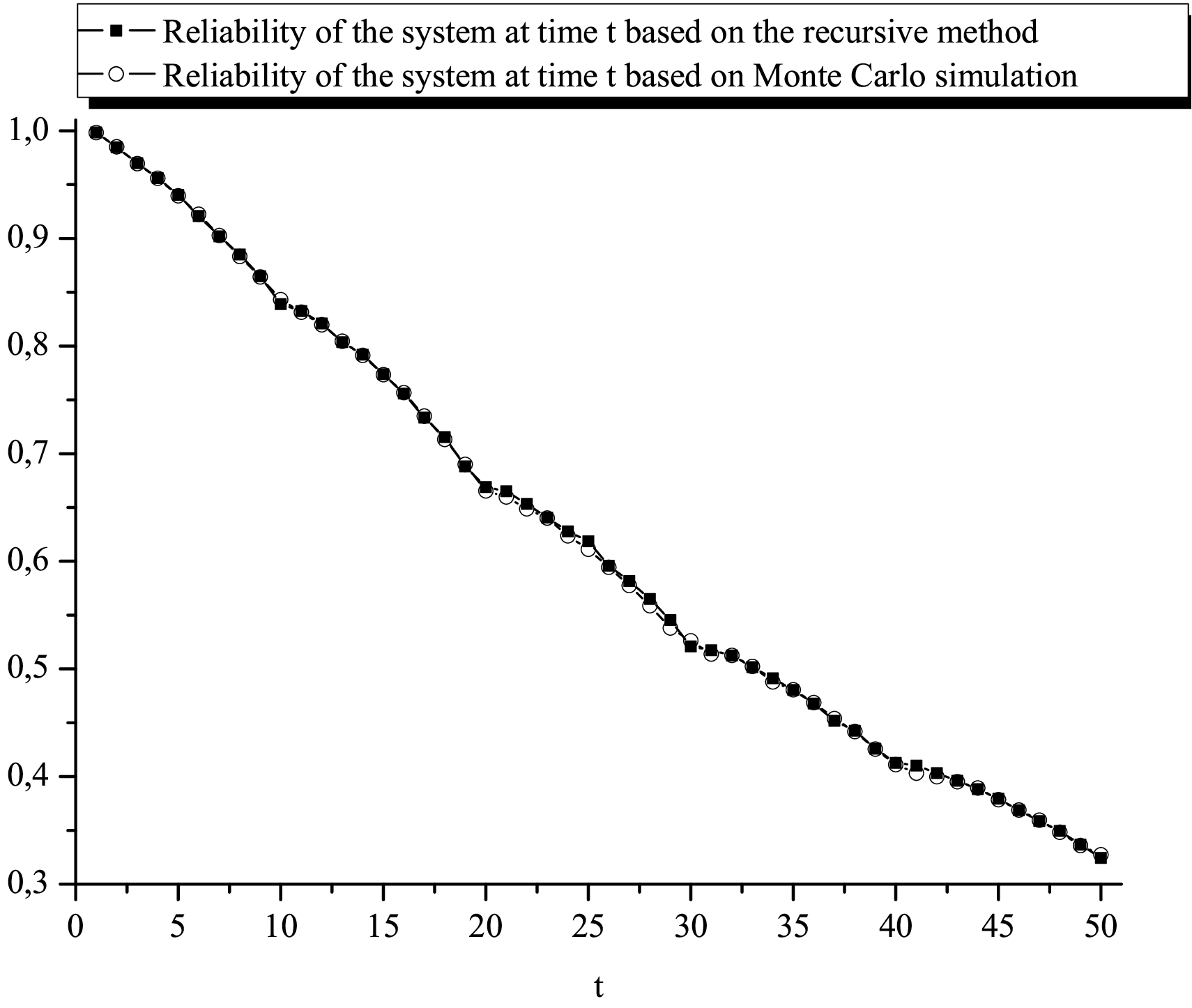}
\caption{Reliability of the system for different values of $t$.}\label{reliabilityFig1}
\end{center}
\end{figure}

Finally, the interval reliability of the system based on (\ref{IRt}) is computed throughout the following steps:

\begin{enumerate} \label{procedimiento_recursivo_variandoT}
\item A grid of size $10$ is obtained by discretising the set $[10,30]$ into $10$ equally spaced points from $10$ to $30$ for $t$.

\item For $s=5$, let $\tilde{IR}_{10}^{14}(t,t+s)$ be the system interval reliability estimation for a time between inspections $T=10$ $t.u.$ and a preventive threshold $M=14$ $d.u.$. The interval reliability of the system is calculated by using the recursive formula given in (\ref{IRt}), replacing $P_{R_{1,p}}^{14}(10k)$ and $P_{R_{1}}^{14}(10k)$ by the estimation $\tilde{P}_{R_{1,p}}^{14}(10k)$ and $\tilde{P}_{R_{1}}^{14}(10k)$, respectively, calculated in Step 2 of the procedure detailed in Section \ref{two-dimensional_transient} with initial conditions $\tilde{IR}_{10}^{14}(0,0)=1$ and $\tilde{R}_{10}^{14}(0)=1$. 
\end{enumerate}

Figure \ref{reliabilityFig} shows the interval reliability of the system versus $t$. As we can observe, the results provide for both methods are very similar. Furthermore we can conclude that, for fixed $T=10$ and $M=14$, the probability that the system does not fail in the interval $[t,t+5]$ is, at least, of the $72\%$ for $15 \leq t \leq 35$.
\begin{figure}[h]
\begin{center}
\includegraphics[scale=0.30]{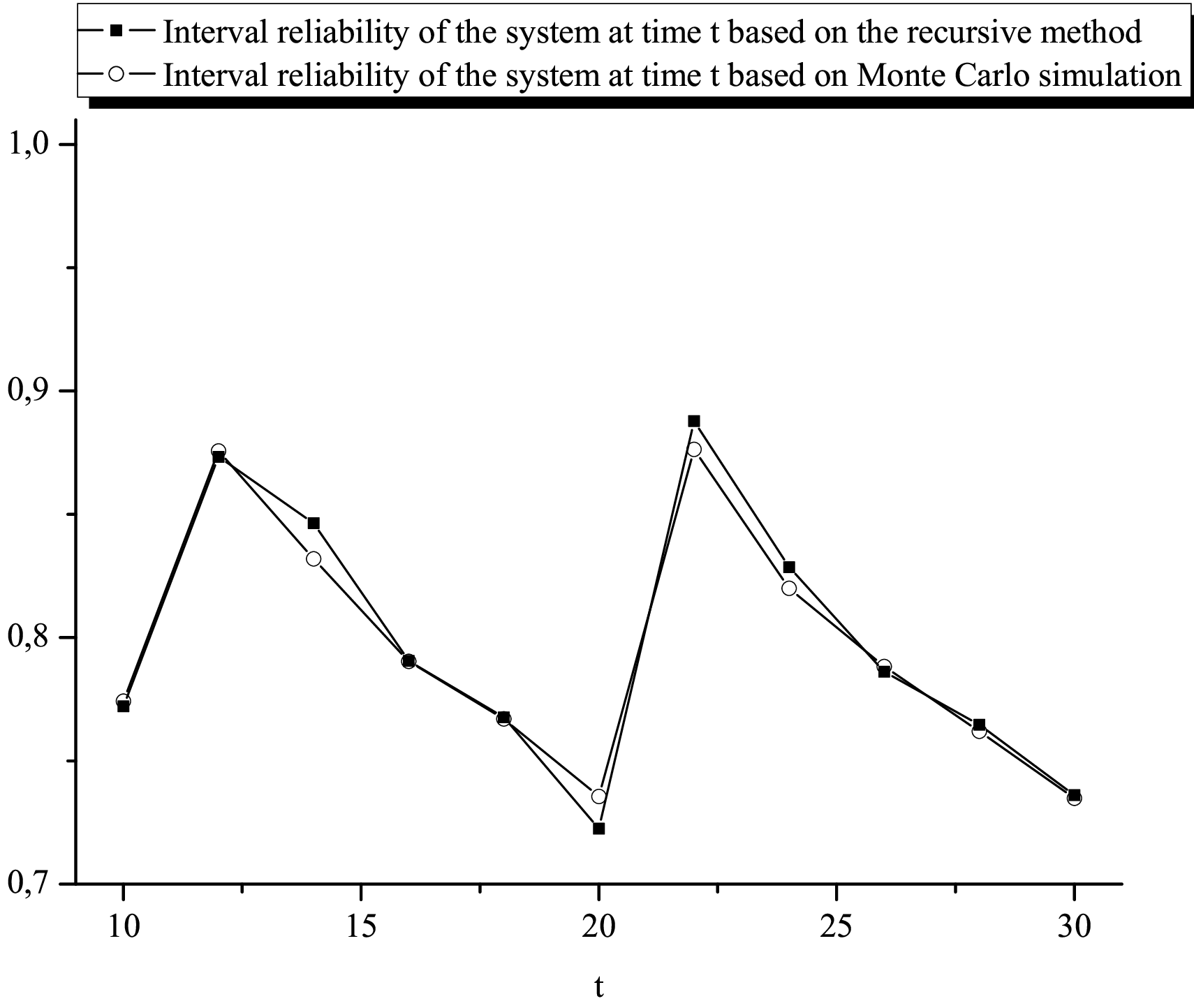}
\caption{Interval reliability of the system for different values of $t$.}\label{reliabilityFig}
\end{center}
\end{figure}

\section{Conclusions and further works}
In this paper, a CBM strategy is analysed by considering a finite life cycle of the system. The system is subject to two different causes of failure, a degradation process modelled under a gamma process and a sudden shock process which follows a DSPP. We consider that both causes of failure are dependent. This dependence is reflected in that the system is more susceptible to external shocks when the deterioration level of the system reaches a certain threshold $M_s$. 

Under these assumptions, the expected cost rate in the life cycle is used as objective function to obtain the optimal maintenance strategy. To this end, a numerical method based on a recursive formula is provided to evaluate the expected cost rate and the standard deviation associated. 

The expected transient cost rate calculated using the recursive formula is compared to both the asymptotic expected cost rate and the expected transient cost rate calculated using strictly Monte Carlo simulation. In addition, the robustness of the gamma process parameters and DSPP is analysed. Furthermore, for the comparison between transient and asymptotic cost rate the results are also provided under a bivariate case, where the time between inspections and the preventive threshold vary simultaneously. In the comparison between the method based on strictly Monte Carlo simulation and the method based on the recursive formula, we observed similar results for a fixed time between inspections $T$. For the preventive threshold $M$, the results presented some differences, which could be explained due to the high variance in the deterioration increments of the degradation process. If the life cycle increases, the recursive method shall tends to be more costly in terms of computation than the method based on strictly Monte Carlo simulation.

Finally, three important performance measures in the maintenance field,  the availability, the reliability and the interval reliability of the system, are analysed under the bivariate optimal maintenance strategy. Recursive formulas are given to obtain these performance measures. Furthermore, the results obtained using the recursive formulation is compared to the results obtained using strictly Monte Carlo simulation.

In this paper, we consider a system subject to a unique degradation process. However, sometimes the system is subject to multiple degradation processes. A possible further extension of this work is to consider a system subject to multiple degradation processes. With respect to the two types of failure, the analysis of this model is based on the dependence of the degradation level of the system on the intensity of shocks. An interesting extension could be to assume a bidirectional relation of dependence where the process affects also to the degradation process. 

\section*{Appendix A}
For $t<T$, $E\left[C_T^M(t)\right]$, is given by
\begin{equation*}\label{Proof EC0t1}
\begin{array}{l}
\begin{aligned}
E\left[C_T^M(t)\right]=&C_d E\left[(t-Y) \mathbf{1}_{\{\sigma_{M_s}<Y<t,~Y<\sigma_L\}}\right]\\
+&C_d E\left[(t-\sigma_L) \mathbf{1}_{\{\sigma_{M_s}<\sigma_L<t,~\sigma_L<Y\}}\right]\\
+&C_d E\left[(t-Y) \mathbf{1}_{\{Y<t,~Y<\sigma_{M_s}\}}\right].\\
\end{aligned}
\end{array}
\end{equation*}

That is
\begin{equation*}\label{Proof EC0t2}
\begin{array}{l}
\begin{aligned}
E\left[C_T^M(t)\right]=&C_d\displaystyle \int_{0}^{t}f_{\sigma_{M_s}}(u)\int_{u}^{t}\left[-\frac{\partial}{\partial v}I(u,v)\right]\\
&\bar F_{\sigma_{L}-\sigma_{M_s}}(v-u)(t-v)~dv~du\\
+&C_d\displaystyle \int_{0}^{t}f_{\sigma_{M_s}}(u)\int_{u}^{t}I(u,v)\\
&f_{\sigma_{L}-\sigma_{M_s}}(v-u)(t-v)~dv~du\\
+&C_d\displaystyle \int_{0}^{t}f_1(u)\bar F_{\sigma_{M_s}}(u)(t-u)du.\\
\end{aligned}
\end{array}
\end{equation*}

For $t\geq T$, $E\left[C_{T}^M(t)\right]$ is conditioned to $R_1$
$$
\begin{array}{l}
\begin{aligned}
E\left[C_{T}^M(t)\right]=&E\left[C_{T}^M(t),R_1\leq t\right]+E\left[C_{T}^M(t),R_1> t\right].
\end{aligned}
\end{array}
$$
Thus, if $R_1>t$
$$
\begin{array}{l}
\begin{aligned}
E&\left[C_{T}^M(t),R_1> t\right]\\
=&\lfloor t/T\rfloor C_I\Bigg(1-\displaystyle\sum_{k=1}^{\lfloor t/T\rfloor}P_{R_1}^M(kT)\Bigg)\\
+&C_dE\left[W_{T}^{M}(\lfloor t/T\rfloor T,t)\right]\Bigg(1-\displaystyle\sum_{k=1}^{\lfloor t/T\rfloor}P_{R_1}^M(kT)\Bigg).
\end{aligned}
\end{array}$$
If $R_1\leq t$, $E\left[C_{T}^M(t)\right]$ can be split into two terms: the cost in the first renewal cycle $\left(C_T^M(R_1)\right)$ and the cost in the remaining time horizon $\left(C_T^M(R_1,t)\right)$. Since $C_T^M(R_1)$ and $C_T^M(R_1,t)$ are independent, we get
\begin{equation*}
\begin{array}{l}
\begin{aligned}
E\left[C_{T}^M(t),R_1\leq t\right]=&E\left[C_T^M(R_1),R_1\leq t\right]\\
+&E\left[C_{T}^M(R_1,t),R_1\leq t\right].
\end{aligned}
\end{array}
\end{equation*}

For a fixed $T$
\begin{equation*}
\begin{array}{l}
\begin{aligned}
E&\left[C_T^M(R_1),R_1\leq t\right]\\
=&\displaystyle\sum_{k=1}^{\lfloor t/T\rfloor}E\left[C_T^M(R_1),R_{1,c}=kT\right]\\
+&\displaystyle\sum_{k=1}^{\lfloor t/T\rfloor}E\left[C_T^M(R_1),R_{1,p}=kT\right],
\end{aligned}
\end{array}
\end{equation*}
being
\begin{equation*}
\begin{array}{l}
\begin{aligned}
E&\left[C_T^M(R_1),R_{1,c}=kT\right]\\
=&\left(C_c+C_I(k-1)\right)P_{R_{1,c}}^M(kT)\\
+&C_d E\left[W_{T}^{M}((k-1)T,kT)\right]P_{R_{1,c}}^M(kT),
\end{aligned}
\end{array}
\end{equation*}
and
\begin{equation*}
E\left[C_T^M(R_1),R_{1,p}=kT\right]=\left(C_p+C_I(k-1)\right)P_{R_{1,p}}^M(kT).
\end{equation*}

Since $C_T^M(R_1,t)$ is stochastically the same as $C_T^M(t-R_1)$, 
$$E\left[C_T^M(R_1,t),R_1=kT\right]=E\left[C_T^M(t-kT)\right]P_{R_1}^M(kT).$$

Hence, $E\left[C_{T}^M(t)\right]$ verifies the following recursive equation
$$E\left[C_{T}^M(t)\right]=\displaystyle \sum_{k=1}^{\lfloor t/T\rfloor}E\left[C_{T}^M(t-kT)\right]P_{R_1}^M(kT)+G_T^M(t),$$
being
\begin{equation*}
\begin{array}{l}
\begin{aligned}
G_T^M(t)=&\displaystyle \sum_{k=1}^{\lfloor t/T\rfloor}\Big(C_p+C_I(k-1)\Big)P_{R_{1,p}}^M(kT)\\
+&\displaystyle \sum_{k=1}^{\lfloor t/T\rfloor}\Big(C_c+C_I(k-1)\Big)P_{R_{1,c}}^M(kT)\\
+&\displaystyle \sum_{k=1}^{\lfloor t/T\rfloor}C_dE\left[W_{T}^M((k-1)T,kT)\right]P_{R_{1,c}}^M(kT)\\
+&\lfloor t/T\rfloor C_I\Big(1-\sum_{k=1}^{\lfloor t/T\rfloor} P_{R_1}^M(kT)\Big)\\
+&C_dE\left[W_T^M(\lfloor t/T\rfloor T,t)\right]\Big(1-\sum_{k=1}^{\lfloor t/T\rfloor} P_{R_1}^M(kT)\Big),\\
\end{aligned}
\end{array}
\end{equation*}
and the result holds.

\section*{Appendix B}
For $t<T$, the expected square cost, $E\left[C_T^M(t)^2\right]$, is given by
\begin{equation*}\label{Proof B0t1}
\begin{array}{l}
\begin{aligned}
E\left[C_T^M(t)^2\right]=&C_d^2 E\left[(t-Y)^2 \mathbf{1}_{\{\sigma_{M_s}<Y<t,~Y<\sigma_L\}}\right]\\
+&C_d^2 E\left[(t-\sigma_L)^2 \mathbf{1}_{\{\sigma_{M_s}<\sigma_L<t,~\sigma_L<Y\}}\right]\\
+&C_d^2 E\left[(t-Y)^2 \mathbf{1}_{\{Y<t,~Y<\sigma_{M_s}\}}\right].\\
\end{aligned}
\end{array}
\end{equation*}

That is
\begin{equation*}\label{Proof B0t2}
\begin{array}{l}
\begin{aligned}
E\left[C_T^M(t)\right]=&C_d^2\displaystyle \int_{0}^{t}f_{\sigma_{M_s}}(u)\int_{u}^{t}\left[-\frac{\partial}{\partial v}I(u,v)\right]\\
&\bar F_{\sigma_{L}-\sigma_{M_s}}(v-u)(t-v)^2~dv~du\\
+&C_d^2\displaystyle \int_{0}^{t}f_{\sigma_{M_s}}(u)\int_{u}^{t}I(u,v)\\
&f_{\sigma_{L}-\sigma_{M_s}}(v-u)(t-v)^2~dv~du\\
+&C_d^2\displaystyle \int_{0}^{t}f_1(u)\bar F_{\sigma_{M_s}}(u)(t-u)^2du.\\
\end{aligned}
\end{array}
\end{equation*}

For $t\geq T$, $E\Big[C_T^{M}(t)^2\Big]$ is conditioned to $R_1$
$$
\begin{array}{l}
\begin{aligned}
E\left[C_{T}^M(t)^2\right]=&E\left[C_{T}^M(t)^2,R_1\leq t\right]+E\left[C_{T}^M(t)^2,R_1> t\right].
\end{aligned}
\end{array}
$$
Hence,
$$
\begin{array}{l}
\begin{aligned}
E&\left[C_{T}^M(t)^2,R_1> t\right]\\
=&\Big(\lfloor t/T\rfloor C_I +C_dE\left[W_{T}^{M}(\lfloor t/T\rfloor T,t)\right]\Big)^2\\
&\Bigg(1-\displaystyle\sum_{k=1}^{\lfloor t/T\rfloor}P_{R_1}^M(kT)\Bigg),\end{aligned}
\end{array}$$
where $P_{R_1}^M$ is given by (\ref{prob_reemplazamiento}). On the other hand
$$
\begin{array}{l}
\begin{aligned}
E&\left[C_{T}^M(t)^2,R_1\leq t\right]\\
=&E\left[\left(C_T^M(R_1)+C_T^M(R_1,t)\right)^2,R_1\leq t\right].
\end{aligned}
\end{array}
$$

Since $C_T^M(R_1)$ and $C_T^M(R_1,t)$ are independent
$$
\begin{array}{l}
\begin{aligned}
E&\left[\left(C_T^M(R_1)+C_T^M(R_1,t)\right)^2,R_1\leq t\right]\\
=&E\left[C_T^M(R_1)^2,R_1\leq t\right]\\
+&E\left[C_{T}^M(R_1,t)^2,R_1\leq t\right]\\
+&2~E\left[C_T^M(R_1),R_1\leq t\right]E\left[C_{T}^M(R_1,t),R_1\leq t\right].
\end{aligned}
\end{array}
$$

Thus,
\begin{equation*}
\begin{array}{l}
\begin{aligned}
E&\Big[C_T^M(R_1)^2,R_1\leq t\Big]\\
=&\displaystyle\sum_{k=1}^{\lfloor t/T\rfloor}E\left[C_T^M(R_1)^2,R_{1,c}=kT\right]\\
+&\displaystyle\sum_{k=1}^{\lfloor t/T\rfloor}E\left[C_T^M(R_1)^2,R_{1,p}=kT\right],
\end{aligned}
\end{array}
\end{equation*}
being
\begin{equation*}
\begin{array}{l}
\begin{aligned}
E&\left[\left(C_T^M(R_1)\right)^2,R_{1,c}=kT\right]\\
=&\left(C_c+C_I(k-1)+C_dE\left[W_{T}^{M}((k-1)T,kT)\right]\right)^2P_{R_{1,c}}^M(kT),
\end{aligned}
\end{array}
\end{equation*}
and
\begin{equation*}
\begin{array}{l}
\begin{aligned}
E\left[\left(C_T^M(R_1)\right)^2,R_{1,p}=kT\right]=\left(C_c+C_I(k-1)\right)^2P_{R_{1,p}}^M(kT).
\end{aligned}
\end{array}
\end{equation*}

Following the same reasoning as in Appendix A, 
\begin{equation*}
\begin{array}{l}
\begin{aligned}
E&\left[C_T^M(R_1),R_1\leq t\right]\\
=&\displaystyle\sum_{k=1}^{\lfloor t/T\rfloor}\left(C_c+C_I(k-1)\right)P_{R_{1,c}}^M(kT)\\
+&\displaystyle\sum_{k=1}^{\lfloor t/T\rfloor}C_d E\left[W_{T}^{M}((k-1)T,kT)\right]P_{R_{1,c}}^M(kT)\\
+&\displaystyle\sum_{k=1}^{\lfloor t/T\rfloor}\left(C_p+C_I(k-1)\right)P_{R_{1,p}}^M(kT).\\
\end{aligned}
\end{array}
\end{equation*}

Since $C_T^M(R_1,t)$ is stochastically the same as $C_T^M(t-R_1)$, we get
$$\begin{array}{l}
\begin{aligned}
E\left[C_T^M(R_1,t),R_1\leq t\right]=&\displaystyle \sum_{k=1}^{\lfloor t/T\rfloor} E\left[C_{T}^M(t-kT)\right]P_{R_1}^M(kT).
\end{aligned}
\end{array}$$ 
Thus,
$$
\begin{array}{l}
\begin{aligned}
E&\left[C_T^M(R_1),R_1\leq t\right]E\left[C_{T}^M(R_1,t),R_1\leq t\right]\\
=&\sum_{k=1}^{\lfloor t/T \rfloor }\left(C_c+C_I(k-1)\right)E\left[C_{T}^M(t-kT)\right]P_{R_{1,c}}^M(kT)\\
+&\sum_{k=1}^{\lfloor t/T \rfloor }C_dE\left[W_{T}^{M}((k-1)T,kT)\right]E\left[C_{T}^M(t-kT)\right]P_{R_{1,c}}^M(kT)\\
+&\sum_{k=1}^{\lfloor t/T \rfloor }\left(C_p+C_I(k-1)\right)E\left[C_{T}^M(t-kT)\right]P_{R_{1,p}}^M(kT).
\end{aligned}
\end{array}
$$
Then, $E\left[C_{T}^M(t)^2\right]$ verifies the following recursive equation
$$E\left[C_{T}^M(t)^2\right]=\displaystyle \sum_{k=1}^{\lfloor t/T\rfloor}E\left[C_{T}^M(t-kT)^2\right]P_{R_1}^M(kT)+H_T^M(t),$$
being
\begin{equation*}
\begin{array}{l}
\begin{aligned}
H_{T}^M(t)=&\displaystyle \sum_{k=1}^{\lfloor t/T \rfloor } \left(C_p+C_I(k-1)\right)^2P_{R_{1,p}}^M(kT)\\
+&\displaystyle \sum_{k=1}^{\lfloor t/T \rfloor } \left(C_c+C_I(k-1)+C_dE\left[W_{T}^M((k-1)T,kT)\right]\right)^2\\
&P_{R_{1,c}}^M(kT)\\
+&2\sum_{k=1}^{\lfloor t/T \rfloor }\left(C_c+C_I(k-1)\right)E\left[C_T^M(t-kT)\right]\\
&P_{R_{1,c}}^M(kT)\\
+&2\sum_{k=1}^{\lfloor t/T \rfloor }C_dE\left[W_{T}^M((k-1)T,kT)\right]\\
&E\left[C_T^M(t-kT)\right]P_{R_{1,c}}^M(kT)\\
+&2\sum_{k=1}^{\lfloor t/T \rfloor }\left(C_p+C_I(k-1)\right)E\left[C_T^M(t-kT)\right]\\
&P_{R_{1,p}}^M(kT)\\
+&\Big(\lfloor t/T\rfloor C_I +C_dE\left[W_T^M(\lfloor t/T\rfloor T,t)\right] \Big)^2\\
&\Bigg(1-\sum_{k=1}^{\lfloor t/T\rfloor}P_{R_1}^M(kT)\Bigg),\\
\end{aligned}
\end{array}
\end{equation*}
and the result holds.

\section*{Appendix C}
For $t<T$, $A_T^M(t)$ is given by
\begin{equation*}\label{Proof A0t1}
\begin{array}{l}
\begin{aligned}
A_T^M(t)=& P\left[t<\sigma_{M_s},~Y>t\right]
+P\left[\sigma_{M_s}<t<\sigma_L,~Y>t\right]\\
=&\bar F_{\sigma_{M_s}}(t)\bar F_{1}(t)+\displaystyle \int_{0}^{t}f_{\sigma_{M_s}}(u)\bar F_{\sigma_{L}-\sigma_{M_s}}(t-u)I(u,t)du.\\
\end{aligned}
\end{array}
\end{equation*}

For $t\geq T$, $A_T^M(t)$ is conditioned to the time to the first replacement
\begin{equation*}
\begin{array}{l}
\begin{aligned}
A_T^M(t)=& \displaystyle\sum_{j=0}^{\infty}\mathbf{1}_{\{R_{j}\leq t<R_{j+1}\}}\\
&\Big[P\left[O(t)<L,~Y>(t-R_j),~R_1\leq t\right]\\
+&P\left[O(t)<L,~Y>(t-R_j),~R_1> t\right]\Big].
\end{aligned}
\end{array}
\end{equation*}

If $R_1>t$
\begin{equation*}
\begin{array}{l}
\begin{aligned}
A_T^M(t)=& \Big[P\left[t<\sigma_{M},~Y>t\right]\\
&+P\left[\lfloor t/T\rfloor T<\sigma_{M}<\sigma_{M_s}< t<\sigma_{L},~Y>t\right]\\
&+P\left[\lfloor t/T\rfloor T<\sigma_{M}< t<\sigma_{M_s},~Y>t\right]\Big]\mathbf{1}_{\left\{M\leq M_s\right\}}\\
+&\Big[P\left[t<\sigma_{M_s},~Y>t\right]\\
&+P\left[\sigma_{M_s}< \lfloor t/T\rfloor T<\sigma_{M}< t<\sigma_{L},~Y>t\right]\\
&+P\left[\sigma_{M_s}<\lfloor t/T\rfloor T< t<\sigma_{M},~Y>t\right]\\
&+P\left[\lfloor t/T\rfloor T<\sigma_{M_s}< t<\sigma_{L},~Y>t\right]\Big]\mathbf{1}_{\left\{M> M_s\right\}}.
\end{aligned}
\end{array}
\end{equation*}

That is
\begin{equation*}
\begin{array}{l}
\begin{aligned}
A_T^M(t)=& \Big[\bar F_{\sigma_{M}}(t)\bar F_{1}(t)\\
&+\displaystyle \int_{\lfloor t/T \rfloor T}^t f_{\sigma_{M}}(u)\int_{u}^{t}f_{\sigma_{M_s}-\sigma_{M}}(v-u)\\
&\bar F_{\sigma_{L}-\sigma_{M_s}}(t-v)I(v,t)~dv~du\\
&+\displaystyle \int_{\lfloor t/T \rfloor T}^t f_{\sigma_{M}}(u)\bar F_{\sigma_{M_s}-\sigma_{M}}(t-u)\bar F_{1}(t)~du\Big]\mathbf{1}_{\left\{M\leq M_s\right\}}\\
+&\Big[\bar F_{\sigma_{M_s}}(t)\bar F_{1}(t)\\
&+\displaystyle \int_{0}^{\lfloor t/T \rfloor T} f_{\sigma_{M_s}}(u)\int_{\lfloor t/T \rfloor T}^{t}f_{\sigma_{M}-\sigma_{M_s}}(v-u)\\
&\bar F_{\sigma_{L}-\sigma_{M}}(t-v)I(u,t)~dv~du\\
&+\displaystyle \int_{0}^{\lfloor t/T \rfloor T} f_{\sigma_{M_s}}(u)\bar F_{\sigma_{M}-\sigma_{M_s}}(t-u)I(u,t)~du\\
&+\displaystyle \int_{\lfloor t/T \rfloor T}^t f_{\sigma_{M_s}}(u)\bar F_{\sigma_{L}-\sigma_{M_s}}(t-u)I(u,t)~du\Big]\mathbf{1}_{\left\{M> M_s\right\}}\\
&=J_{T,1}^M(t)\mathbf{1}_{\left\{M\leq M_s\right\}}+J_{T,2}^M(t)\mathbf{1}_{\left\{M> M_s\right\}}.
\end{aligned}
\end{array}
\end{equation*}

If $R_1\leq t$, 
\begin{equation*}
\begin{array}{l}
\begin{aligned}
\displaystyle\sum_{j=0}^{\infty}&\mathbf{1}_{\{R_{j}\leq t<R_{j+1}\}}P\left[O(t)<L,~Y>(t-R_j),~R_1\leq t\right]\\
=\displaystyle \sum_{k=1}^{\lfloor t/T\rfloor}&P_{R_1}^M(kT)\Bigg[ \sum_{j=0}^{\infty}\mathbf{1}_{\{R_{j}\leq t<R_{j+1}\}}\\
&P\left[O(t-kT)<L,~Y>(t-kT-R_j)\right]\Bigg]\\
=\displaystyle \sum_{k=1}^{\lfloor t/T\rfloor}&A_T^M(t-kT)P_{R_1}^M(kT).
\end{aligned}
\end{array}
\end{equation*}

Then, for $t\geq T$, $A_T^M(t)$ verifies the following recursive equation
\begin{equation*}
\begin{array}{l}
\begin{aligned}
A_T^M(t)=&\displaystyle \sum_{k=1}^{\lfloor t/T\rfloor}A_T^M(t-kT)P_{R_1}^M(kT)\\
+&J_{T,1}^M(t)\mathbf{1}_{\left\{M\leq M_s\right\}}+J_{T,2}^M(t)\mathbf{1}_{\left\{M> M_s\right\}},
\end{aligned}
\end{array}
\end{equation*}
and the result holds.

\section*{Appendix D}
For $t<T$, thee is no maintenance action on $[0,t]$, hence $R_T^M(t)$ is equal to $A_T^M(t)$. 

For $t\geq T$, $R_T^M(t)$ is conditioned to the time of the first replacement
\begin{equation*}
\begin{array}{l}
\begin{aligned}
R_T^M(t)=&P\left[O(u)<L,~\forall u\in (0,t],~N_s(0,t)=0,~R_{1}\leq t\right]\\
+&P\left[O(u)<L,~\forall u\in (0,t],~N_s(0,t)=0,~R_{1}> t\right]. 
\end{aligned}
\end{array}
\end{equation*}

If $R_{1}>t$
\begin{eqnarray*}
R_T^M(t)& =& A_T^M(t) \\ &=& J_{T,1}^M(t)\mathbf{1}_{\left\{M\leq M_s\right\}}+J_{T,2}^M(t)\mathbf{1}_{\left\{M> M_s\right\}}.
\end{eqnarray*}

If $R_{1}\leq t$, 
\begin{equation*}
\begin{array}{l}
\begin{aligned}
R_T^M(t)=&P\left[O(u)<L,~\forall u\in (0,t],N_s(0,t)=0,~R_{1}\leq t\right]\\
=&\displaystyle \sum_{k=1}^{\lfloor t/T\rfloor}P_{R_{1,p}}^M(kT)P\left[O(u-kT)<L,\right.\\
&\left.\forall u\in (0,t-kT],N_s(0,t-kT)=0\right]\\
=&\displaystyle \sum_{k=1}^{\lfloor t/T\rfloor}R_T^M(t-kT)P_{R_{1,p}}^M(kT).
\end{aligned}
\end{array}
\end{equation*}

Then, for $t\geq T$, $R_T^M(t)$ verifies the following recursive equation
\begin{equation*}
\begin{array}{l}
\begin{aligned}
R_T^M(t)=&\displaystyle \sum_{k=1}^{\lfloor t/T\rfloor}R_T^M(t-kT)P_{R_{1,p}}^M(kT)\\
+&J_{T,1}^M(t)\mathbf{1}_{\left\{M\leq M_s\right\}}+J_{T,2}^M(t)\mathbf{1}_{\left\{M> M_s\right\}},
\end{aligned}
\end{array}
\end{equation*}
and the result holds.

\section*{Appendix E}
For $(t+s)<T$, there is no maintenance action on $[0, t+s]$, hence $IR_T^M(t,t+s)$ is equal to $R_T^M(t+s)$. 
For $t+s\geq T$, $IR_T^M(t,t+s)$ is conditioned to the time of the first replacement
\begin{equation*}
\begin{array}{l}
\begin{aligned}
IR_T^M(t,t+s)=&P\left[O(u)<L,\forall u\in (t,t+s],~\right.\\
&\left.N_s(t,t+s)=0,~R_{1}\leq t\right]\\
+&P\left[O(u)<L,~\forall u\in (t,t+s],\right.\\
&\left.N_s(t,t+s)=0,~t<R_1< t+s\right]\\
+&P\left[O(u)<L,~\forall u\in (t,t+s],\right.\\
&\left.N_s(t,t+s)=0,~R_1\geq t+s\right]\Big].
\end{aligned}
\end{array}
\end{equation*}
If $R_1\geq t+s$
\begin{equation*}
\begin{array}{l}
\begin{aligned}
IR_T^M(t,t+s)=&A_T^M(t+s)\\
=&J_{T,1}^M(t+s)\mathbf{1}_{\left\{M\leq M_s\right\}}+J_{T,2}^M(t+s)\mathbf{1}_{\left\{M> M_s\right\}}.
\end{aligned}
\end{array}
\end{equation*}

If $t<R_1<t+s$
\begin{equation*}
\begin{array}{l}
\begin{aligned}
IR_T^M(t,t+s)=&P\left[O(u)<L,~\forall u\in (t,t+s],\right.\\
&\left.N_s(t,t+s)=0,~t<R_{1}<t+s\right]\\
=&\displaystyle \sum_{k=\lfloor t/T\rfloor+1}^{\lfloor (t+s)/T\rfloor}P_{R_{1,p}}^M(kT)P\left[O(u-kT)<L,\right.\\
&\left.\forall u\in (0,t+s-kT],N_s(0,t+s-kT)=0\right]\\
=&\displaystyle \sum_{k=\lfloor t/T\rfloor+1}^{\lfloor (t+s)/T\rfloor}R_T^M(t+s-kT)P_{R_{1,p}}^M(kT).
\end{aligned}
\end{array}
\end{equation*}

If $R_1\leq t$
\begin{equation*}
\begin{array}{l}
\begin{aligned}
IR_T^M(t,t+s)=&P\left[O(u)<L,\forall u\in (t,t+s],\right.\\
&\left.~N_s(t,t+s)=0,~R_1\leq t\right]\\
=&\displaystyle \sum_{k=1}^{\lfloor t/T\rfloor}P_{R_1}^M(kT)P\left[O(u-kT)<L,\right.\\
&\left.\forall u\in (t-kT,t+s-kT],\right.\\
&\left.N_s(t-kT,t+s-kT)=0\right]\\
=&\displaystyle \sum_{k=1}^{\lfloor t/T\rfloor}IR_T^M(t-kT,t+s-kT)P_{R_{1}}^M(kT).
\end{aligned}
\end{array}
\end{equation*}

Then, for $t+s\geq T$,  $IR_T^M(t,t+s)$ verifies the following recursive equation
\begin{equation*}
\begin{array}{l}
\begin{aligned}
IR_T^M(t,t+s)=&\displaystyle \sum_{k=\lfloor t/T\rfloor+1}^{\lfloor (t+s)/T\rfloor}R_T^M(t+s-kT)P_{R_{1,p}}^M(kT)\\
+&\displaystyle \sum_{k=1}^{\lfloor t/T\rfloor}IR_T^M(t-kT,t+s-kT)P_{R_{1}}^M(kT)\\
+&J_{T,1}^M(t+s)\mathbf{1}_{\left\{M\leq M_s\right\}}+J_{T,2}^M(t+s)\mathbf{1}_{\left\{M> M_s\right\}},
\end{aligned}
\end{array}
\end{equation*}
and the result holds.

\section*{Acknowledgements}
This research was supported by \textit{Ministerio de
Econom{\'\i}a y Competitividad, Spain} (Project MTM2015-63978-P), \textit{Gobierno de Extremadura, Spain} (Project {GR15106}), and \textit{European
Union} (European Regional Development Funds). Funding for a PhD
grant comes from \textit{Fundaci\'{o}n Valhondo} (Spain).

\end{document}